\documentclass[12pt]{article}
\usepackage{graphicx}
\usepackage{tikz}
\usepackage{amssymb}
\usepackage{mathrsfs}
\usepackage{amsmath}
\usepackage{amsfonts}
\usepackage{amsthm}
\usepackage[english]{babel}
\usepackage{mdwlist}
\usepackage{comment}

\usepgflibrary{decorations.markings}
\usetikzlibrary{decorations.markings}
\usetikzlibrary{patterns}
\DeclareMathOperator{\im}{Im}
\DeclareMathOperator{\re}{Re}

\DeclareMathOperator{\diag}{diag}
\DeclareMathOperator{\transpose}{T}
\DeclareMathOperator{\sign}{sgn}
\newcommand{\inner}[1]{\ensuremath{\stackrel{\circ}{#1}}}

\newtheorem{theorem}{Theorem}
\newtheorem{lemma}[theorem]{Lemma}
\newtheorem{proposition}[theorem]{Proposition}
\theoremstyle{definition}

\theoremstyle{remark}

\newtheorem{remark}[theorem]{Remark}
\numberwithin{equation}{section}
\numberwithin{theorem}{section}
\DeclareMathOperator{\Ang}{\textnormal{Ang}}

\def\barroman#1{\sbox0{#1}\dimen0=\dimexpr\wd0+1pt\relax
  \makebox[\dimen0]{\rlap{\vrule width\dimen0 height 0.06ex depth 0.06ex}%
    \rlap{\vrule width\dimen0 height\dimexpr\ht0+0.03ex\relax 
            depth\dimexpr-\ht0+0.09ex\relax}%
    \kern.5pt#1\kern.5pt}}
\definecolor{verylightgray}{rgb}{.9,.9,.9}

\newdimen \arrowsize
\pgfarrowsdeclare{stealth0}{stealth0}
{
\arrowsize=0.28pt
\advance \arrowsize  by.3\pgflinewidth%
\pgfarrowsleftextend{-8\arrowsize}
\pgfarrowsrightextend{0\arrowsize}
}
{
\arrowsize=0.28pt
\advance\arrowsize by.3\pgflinewidth%
\pgfpathmoveto{\pgfqpoint{4\arrowsize}{0pt}}
\pgfpathlineto{\pgfqpoint{-4\arrowsize}{4\arrowsize}}
\pgfpathlineto{\pgfpoint{-\arrowsize}{0}}
\pgfpathlineto{\pgfqpoint{-4\arrowsize}{-4\arrowsize}}
\pgfusepathqfill
}

\title{Critical behavior in Angelesco ensembles}
\author{K. Deschout and A.B.J. Kuijlaars}
\date{\today}

\begin{document}

\maketitle

\begin{center}
    Department of Mathematics, Katholieke Universiteit Leuven, \\
    Celestijnenlaan 200B, 3001 Leuven, Belgium, \\
klaas.deschout\symbol{'100}gmail.com, \
arno.kuijlaars\symbol{'100}wis.kuleuven.be
\end{center}

\begin{abstract}
We consider Angelesco ensembles with respect to two modified Jacobi weights on 
touching intervals $[a,0]$ and $[0,1]$, for $a < 0$.  As $a \to -1$ the particles 
around $0$ experience a phase transition.  This transition is studied in a double 
scaling limit, where we let the number of particles of the ensemble tend to infinity 
while the parameter $a$ tends to $-1$ at a rate of $\mathcal{O}(n^{-1/2})$.  The 
correlation kernel converges, in this regime, to a new kind of universal kernel, 
the Angelesco kernel $\mathbb{K}^{\Ang}$.  The result follows from the Deift/Zhou 
steepest descent analysis, applied to the Riemann-Hilbert problem for multiple orthogonal polynomials.
\end{abstract}

\section{Introduction and statement of results}

Multiple orthogonal polynomial (MOP) ensembles \cite{kui10} form an extension of the more 
familiar orthogonal polynomial (OP) ensembles \cite{kon04}.  The latter appear as 
the eigenvalue distributions of unitary random matrix ensembles.

A major topic of interest in random matrix theory deals with the universality conjecture.  
This conjecture claims that, as the size of the matrices tends to infinity, the local 
eigenvalue statistics of a random matrix ensemble converge to universal limits, i.e., 
limits that are independent of the precise probability distribution on the matrices. The
limits only depend on macroscopic eigenvalue characteristics and on the symmetry class
of the matrix ensemble.

For the class of unitary random matrix ensembles the problem concerning universality can 
be translated into a problem concerning the associated sequence of orthogonal polynomials (OP).  
By the Gaudin-Mehta method \cite{game60, gau61, meh60} the eigenvalue density function can be 
written as a determinantal point process.  The relevant correlation kernel is the OP kernel, 
which is also known as the Christoffel-Darboux kernel.  Universality is then expressed as the 
convergence of this OP kernel to a certain universal limiting kernel.  This convergence can 
be derived by considering the large degree asymptotics of the OP.  Common limiting kernels 
include the sine, Airy and Bessel kernels.

For other random matrix ensembles a connection exists with multiple orthogonal polynomials (MOP): 
the eigenvalues are a determinantal point process with the so-called MOP kernel
as correlation kernel.  
A prime example of such a connection occurs in the unitary random matrix model with external source.  
By analysis of the MOP kernel in this model universality has been shown and the 
sine and Airy kernels appeared \cite{abk05, blku04b}.  However, in a critical case this model 
also exhibits behavior that does not appear in OP ensembles: under certain scaling 
the MOP kernel converges to the so-called Pearcey kernel \cite{blku07, brhi98, trwi06}.  
MOP ensembles also gave rise to new critical behavior in the papers \cite{duge11, kmfw11}.

This is a recurring feature: MOP ensembles can exhibit a wider variety of critical phenomena 
than OP ensembles.  In this paper we consider another type of MOP ensembles, the 
so-called Angelesco ensembles.  These ensembles do not appear in a natural way as the 
eigenvalue process of a random matrix ensemble, but the notion of universality still makes sense.  
We show that a new kind of critical behavior occurs in Angelesco systems, determined by 
a new kind of universal kernel, that we call the Angelesco kernel $\mathbb{K}^{\Ang}$.

An Angelesco system of weights \cite{ang19,niso91} is a system of $r \geq 2$ weights 
$\vec{w} = (w_{1}, \mathellipsis, w_{r})$ on the real line such that the 
supports of the weights  are contained in intervals with pairwise disjoint interiors.
Thus there exist bounded real intervals $\Delta_{1}, \dots, \Delta_{r}$ such that 
\begin{equation} \label{eq:Deltaj}
 \inner{\Delta}_{j} \cap \inner {\Delta}_{k} = \varnothing  \quad \textrm{ for } j \neq k
\end{equation}
and $w_j$ is non-negative and integrable on the real line with
\begin{equation} \label{eq:wj}
	w_{j}(x) \equiv 0 \quad \textrm{ for } x \in \mathbb R \setminus \Delta_j. 
\end{equation}

Angelesco ensembles are MOP ensembles associated with Angelesco systems of weights.  
Let $\vec{n}$ be a multi-index $(n_{1}, \dots, n_{r}) \in \mathbb{N}^{r}$ and define 
$|\vec{n}| := n_{1}+ \dots + n_{r}$.  The Angelesco ensemble with respect to $\vec{w}$ and $\vec{n}$ 
is the probability measure on $\mathbb R^{|\vec{n}|}$ with probability density function
\begin{equation} \label{eq:Pn}
	\mathcal P_{\vec{n}}(x_{1}, \dots, x_{|\vec{n}|}) := \frac{1}{Z_{\vec{n}}} \det \left[ x_{k}^{j-1} \right]_{j,k=1}^{|\vec{n}|} 
	\cdot \det \left[ f_{j}(x_{k}) \right]_{j,k=1}^{|\vec{n}|},
\end{equation} 
where $Z_{\vec{n}}$ is a normalization constant and 
\begin{equation}
	f_{n_{1}+\dots + n_{i-1} + j}(x) := x^{j-1} w_{i}(x) \quad 
	\textrm{ for } j = 1, \dots, n_i, i = 1, \dots, r.  
\end{equation}
The Angelesco properties \eqref{eq:Deltaj}-\eqref{eq:wj} imply that \eqref{eq:Pn} 
is non-negative for every $(x_1, \mathellipsis, x_{|\vec{n}|}) \in \mathbb R^{|\vec{n}|}$,
and can be non-zero only if $n_j$ of the particles $x_1, \mathellipsis, x_{|\vec{n}|}$ are in $\Delta_j$
for every $j = 1, \mathellipsis, r$. Thus in an  Angelesco ensemble, with probability $1$, $n_j$ particles are
located in $\Delta_j$ for every $j=1, \mathellipsis, r$.

The probability density function \eqref{eq:Pn} is a biorthogonal ensemble \cite{bor99},
which is a special case of a determinantal point process. 
The correlation kernel is
\begin{equation} \label{eq:Kn} 
	K_{\vec{n}}(x,y) = \sum_{j=1}^{|\vec{n}|} \sum_{k=1}^{|\vec{n}|} \left(M^{-1} \right)_{k,j} x^{j-1} f_k(y) 
	\end{equation}
where $(M^{-1})_{k,j}$ is the $kj$-th entry of  the inverse of the matrix
\begin{equation} \label{eq:M} 
	M =  \left( m_{j,k} \right)_{j,k=1}^{|\vec{n}|}, \qquad m_{j,k} = \int x^{j-1} f_k(x) dx. 
	\end{equation}
That is, see  \cite{blku04a, kui10}:
\begin{equation}
	\mathcal P_{\vec{n}}(x_{1}, \dots, x_{|\vec{n}|}) = 
	\frac{1}{|\vec{n}|!} \det \left[ K_{\vec{n}}(x_{j},x_{k} ) \right]_{j,k =1}^{|\vec{n}|},
\end{equation}
and for each $m = 1, \ldots, |\vec{n}|$,
\begin{equation} 
	\rho_m(x_1, \ldots, x_m) = \det \left[ K_{\vec{n}}(x_j,x_k) \right]_{j,k=1}^m
\end{equation}
	where $\rho_m$ is the $m$-point correlation function.
Another representation for $K_{\vec{n}}$ is by means of a determinant
\[ K_{\vec{n}}(x,y) = \frac{-1}{\det M}
	\begin{vmatrix}  M & \begin{smallmatrix} 1 \\ x \\ \vdots \\ x^{|\vec{n}|-1} \end{smallmatrix} \\
		\begin{smallmatrix} f_1(y) & f_2(y) & \dots & f_{|\vec{n}|}(y) \end{smallmatrix} & 0 
		\end{vmatrix}
		\]
where $M$ is the moment matrix \eqref{eq:M}.

In this paper we consider Angelesco ensembles with respect to two weights $w_{1}$ and $w_{2}$.  
Also we restrict ourselves to diagonal multi-indices $\vec{n} = (n,n)$.  
Our interest is in the local asymptotics of the MOP kernel $K_{\vec{n}} = K_{n,n}$ associated 
to the weights $w_{1}$ and $w_{2}$.  By rescaling if necessary we can assume that $\Delta_{1} = [a,b]$ 
and $\Delta_{2} = [0,1]$ for some $a < b \leq 0$.

The type of local behavior around a certain point is suggested by the behavior of the 
limiting mean particle density in that point. Assume that  $w_j > 0$ a.e.\ on $\Delta_j$. Then
this limiting density only depends on the endpoints of the two intervals $\Delta_{1}$ and $\Delta_{2}$.  
Depending on the values of $a$ and $b$ we get a qualitatively different picture.  
We can distinguish 3 cases and a number of phase transitions, see Figure \ref{fig:121}.

\begin{figure}
\centering
\begin{tikzpicture}
\filldraw[color = verylightgray] (0,0) -- (-.38,-3.8) -- (.5,-3.8) -- (.5, .5) -- (-4.8, .5) -- (-4.8,0) -- (0,0);
\draw [postaction = decorate, decoration = {markings, mark = at position 1 with {\arrow{stealth};}}] (-5,0) -- (1,0);
\draw (1,0) node[above right] {$a$};
\draw [postaction = decorate, decoration = {markings, mark = at position 1 with {\arrow{stealth};}}] (0,-4) -- (0,1);
\draw (0,1) node[above left] {$b$};
\draw [line width = .5] (-.38, -3.8) -- (0,0);
\filldraw (0,0) circle (0.03);
\filldraw [color = black] (-2,0) circle (0.03);
\filldraw [color = black] (-.25, -2.5) circle (0.03);
\draw[color = red, line width = 1]  plot[smooth] coordinates{
(-2.000000,0.000000)(-2.100000,-0.000088)(-2.200000,-0.000671)(-2.300000,-0.002160)(-2.400000,-0.004884)(-2.500000,-0.009107)(-2.600000,-0.015038)(-2.700000,-0.022835)(-2.800000,-0.032620)(-2.900000,-0.044481)(-3.000000,-0.058480)(-3.100000,-0.074654)(-3.200000,-0.093023)(-3.300000,-0.113593)(-3.400000,-0.136355)(-3.500000,-0.161290)(-3.600000,-0.188374)(-3.700000,-0.217572)(-3.800000,-0.248848)(-3.900000,-0.282159)(-4.000000,-0.317460)(-4.100000,-0.354705)(-4.200000,-0.393845)(-4.300000,-0.434831)(-4.400000,-0.477612)(-4.500000,-0.522139)(-4.600000,-0.568361)(-4.700000,-0.616230)(-4.800000,-0.665696)(-4.900000,-0.716712)(-5.000000,-0.769231)
};
\draw [color = red, line width = 1] plot [smooth] coordinates{(-2.000000, 0.000000)(-1.900019, -0.000097)(-1.800156, -0.000820)(-1.700539, -0.002916)(-1.601312, -0.007289)(-1.502630, -0.015026)(-1.404664, -0.027435)(-1.307603, -0.046078)(-1.211652, -0.072827)(-1.117036, -0.109912)(-1.024000, -0.160000)(-0.932809, -0.226266)(-0.843750, -0.312500)(-0.757135, -0.423220)(-0.673297, -0.563820)(-0.592593, -0.740741)(-0.515402, -0.961683)(-0.442125, -1.235866)(-0.373178, -1.574344)(-0.308991, -1.990395)(-0.250000, -2.500000)
};
\draw [font = \footnotesize] (-2.5,-2) node[left]{III};
\draw [font = \footnotesize] (-.4,-.5) node {I};
\draw [font = \footnotesize] (-4.8,-.3) node {II}; 
\draw (-2,0) node[below] {$(-1,0)$};
\draw (-.25, -2.5) node[left] {$ \left( -\frac{1}{4}, -\frac{1}{4} \right)$};
\draw [line width = .5](.05,.1) -- (.4,.8);
\draw (.6, .7) node[above] {$(0,0)$};
\end{tikzpicture}
\caption{The phase diagram for Angelesco ensembles with two weights, supported on $[a,b]$ and $[0,1]$.  
The gray areas represent values of $a$ and $b$ that are forbidden because of the restriction $a < b \leq 0$.  
Remark that the scale on the two axes is different.}
\label{fig:121}
\end{figure}
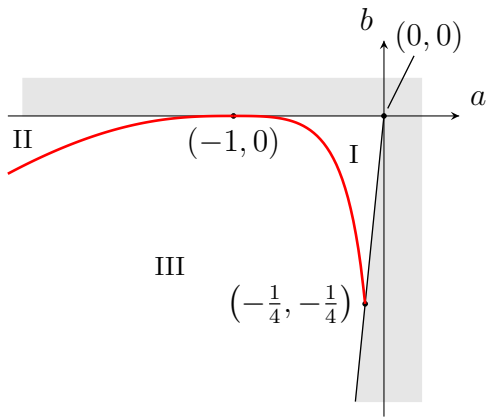

Case {\footnotesize III} corresponds to relatively large separation between the two intervals.  
As a consequence there are hard edges in all four endpoints $a$, $b$, $0$ and $1$, meaning that the 
limiting particle distribution blows up like an inverse square root. 
After rescaling around  one of these endpoints the MOP kernels converge to limiting kernels that
depend on the behavior of the weight at that endpoint. 
In the typical case of Jacobi-type weight functions
\begin{equation} \label{eq:w1w2} 
\begin{aligned} 
	w_1(x) & = (x-a)^{\alpha} (b-x)^{\beta} h_1(x) \chi_{[a,b]}(x), \\
  w_2(x) & = x^{\gamma} (1-x)^{\delta} h_2(x) \chi_{[0,1]}(x), 
  \end{aligned} \end{equation}
this leads to Bessel kernels.

In the cases {\footnotesize I} and {\footnotesize II} the separation is small.  
A gap will emerge on the larger of the two intervals $\Delta_{1}$ and $\Delta_{2}$, due 
to the \textit{pushing of the zeros} effect.  There will then be three hard edges, and 
one soft edge, where the limiting particle density vanishes like a square root.  
Around the soft edge the rescaled MOP kernel will converge to the Airy kernel
(in the case of weights \eqref{eq:w1w2}).

The transition between case {\footnotesize III} and cases {\footnotesize I} and 
{\footnotesize II} is probably related to the Painlev\'{e} II equation as in \cite{clku08}. 
The Bessel, Airy and Painlev\'{e} II kernels already appear as scaling limits for OP ensembles 
\cite{for10}.

This paper deals with the critical point $(a,b) = (-1,0)$, where the two intervals 
$\Delta_{1}$ and $\Delta_{2}$ are touching and of equal size.  
In this case the limiting particle density behaves like $|x|^{-\frac{1}{3}}$ as $x \rightarrow 0$,
see \cite{deku11,kal79}.
We study this critical case in a double scaling limit, where we put $b = 0$ and 
let $a$ tend to $-1$ as $n \rightarrow \infty$:
\begin{equation}
a = a_{n} := -1 + \frac{ \sqrt{2} \tau }{n^{\frac{1}{2}}} \quad \textrm{ for } \tau \in \mathbb{R}.
\end{equation}
The parameter $\tau$ is known as the double scaling parameter.  
This scaling regime corresponds to the phase transition between cases {\footnotesize I} and {\footnotesize II}
with $b$ fixed at $0$.

For $a \neq -1$ and $b =0$ we have a gap in the spectrum between $0$ and $s$ where 
$s = s_{a}$ is given by
\begin{equation}
s_{a} = \frac{ (a+1)^{3}}{9(a^{2}-a+1)},
\end{equation} 
see \cite{kal79}. In $s$ a soft edge appears, while in $0$ we have a hard edge.  
Plots of the limiting particle densities in the cases $a < -1$, $a =-1$ and $a > -1$ 
are given in Figure \ref{fig:187}. As $a$ tends to $-1$ the gap closes and the 
soft edge meets with the hard edge.  This soft-to-hard-edge collision will give rise 
to new critical behavior in $0$.

\begin{figure}[th]
\centering
\begin{tikzpicture}[xscale = 1.3, yscale = 1.3, line width = .5]
\useasboundingbox (-4.5,-1) rectangle (3.7,1);
\draw (2.85,-0.7) node[text = blue] {$a > -1$};
\draw (2.85, -1) node {$(a = -0.7)$};
\draw (0,-0.7) node[text = blue] {$a = -1$};
\draw (-3.25,-0.7) node[text = blue] {$a < -1$};
\draw (-3.25,-1) node {$(a=-1.5)$};
\clip (-6,-1) rectangle (4.5,1);
\draw [color = black] (-1,0) -- (1,0);
\filldraw (-1,0) circle (.02) node[below] {$-1$};
\filldraw (0,0) circle (.02) node[below] {$0$};
\filldraw (1,0) circle (.02)  node[below] {$1$};
\draw[color = red, line width = .5] plot[smooth] coordinates{
 (-0.980000, 1.090445)( -0.960000, 0.776795)( -0.940000, 0.639050)( -0.920000, 0.557697)( -0.900000, 0.502734)( -0.880000, 0.462600)( -0.860000, 0.431774)( -0.840000, 0.407244)( -0.820000, 0.387209)( -0.800000, 0.370518)( -0.780000, 0.356399)( -0.760000, 0.344309)( -0.740000, 0.333857)( -0.720000, 0.324754)( -0.700000, 0.316779)( -0.680000, 0.309763)( -0.660000, 0.303570)( -0.640000, 0.298097)( -0.620000, 0.293257)( -0.600000, 0.288983)( -0.580000, 0.285219)( -0.560000, 0.281919)( -0.540000, 0.279048)( -0.520000, 0.276575)( -0.500000, 0.274478)( -0.480000, 0.272739)( -0.460000, 0.271346)( -0.440000, 0.270291)( -0.420000, 0.269570)( -0.400000, 0.269184)( -0.380000, 0.269139)( -0.360000, 0.269447)( -0.340000, 0.270124)( -0.320000, 0.271194)( -0.300000, 0.272691)( -0.280000, 0.274657)( -0.260000, 0.277151)( -0.240000, 0.280247)( -0.220000, 0.284046)( -0.200000, 0.288686)( -0.180000, 0.294355)( -0.160000, 0.301320)( -0.140000, 0.309977)( -0.120000, 0.320935)( -0.100000, 0.335193)( -0.080000, 0.354543)( -0.060000, 0.382623)( -0.040000, 0.428486)( -0.020000, 0.525878)( -0.018000, 0.543019)( -0.016000, 0.562977)( -0.014000, 0.586659)( -0.012000, 0.615448)( -0.010000, 0.651607)( -0.008000, 0.699155)( -0.006000, 0.766210)( -0.004000, 0.872851)( -0.002000, 1.093359) ( 0.002000, 1.093359)( 0.004000, 0.872851)( 0.006000, 0.766210)( 0.008000, 0.699155)( 0.010000, 0.651607)( 0.012000, 0.615448)( 0.014000, 0.586659)( 0.016000, 0.562977)( 0.018000, 0.543019)( 0.040000, 0.428486)( 0.060000, 0.382623)( 0.080000, 0.354543)( 0.100000, 0.335193)( 0.120000, 0.320935)( 0.140000, 0.309977)( 0.160000, 0.301320)( 0.180000, 0.294355)( 0.200000, 0.288686)( 0.220000, 0.284046)( 0.240000, 0.280247)( 0.260000, 0.277151)( 0.280000, 0.274657)( 0.300000, 0.272691)( 0.320000, 0.271194)( 0.340000, 0.270124)( 0.360000, 0.269447)( 0.380000, 0.269139)( 0.400000, 0.269184)( 0.420000, 0.269570)( 0.440000, 0.270291)( 0.460000, 0.271346)( 0.480000, 0.272739)( 0.500000, 0.274478)( 0.520000, 0.276575)( 0.540000, 0.279048)( 0.560000, 0.281919)( 0.580000, 0.285219)( 0.600000, 0.288983)( 0.620000, 0.293257)( 0.640000, 0.298097)( 0.660000, 0.303570)( 0.680000, 0.309763)( 0.700000, 0.316779)( 0.720000, 0.324754)( 0.740000, 0.333857)( 0.760000, 0.344309)( 0.780000, 0.356399)( 0.800000, 0.370518)( 0.820000, 0.387209)( 0.840000, 0.407244)( 0.860000, 0.431774)( 0.880000, 0.462600)( 0.900000, 0.502734)( 0.920000, 0.557697)( 0.940000, 0.639050)( 0.960000, 0.776795)( 0.980000, 1.090445)
};
\draw[black] (-4.5,0) -- (-2,0);
\filldraw (-4.5,0) circle (.02) node[below] {$a$};
\filldraw (-3,0) circle (.02) node[below] {$0$};
\filldraw (-2,0) circle (.02)  node[below] {$1$};
\draw[color = red, line width = .5] plot[smooth] coordinates{
( -4.486500, 1.058009)( -4.473000, 0.750927)( -4.459500, 0.615444)( -4.446000, 0.535020)( -4.432500, 0.480373)( -4.419000, 0.440216)( -4.405500, 0.409152)( -4.392000, 0.384234)( -4.378500, 0.363698)( -4.365000, 0.346417)( -4.351500, 0.331631)( -4.338000, 0.318807)( -4.324500, 0.307560)( -4.311000, 0.297603)( -4.297500, 0.288717)( -4.284000, 0.280732)( -4.270500, 0.273513)( -4.257000, 0.266953)( -4.243500, 0.260964)( -4.230000, 0.255474)( -4.216500, 0.250424)( -4.203000, 0.245764)( -4.189500, 0.241450)( -4.176000, 0.237448)( -4.162500, 0.233726)( -4.149000, 0.230258)( -4.135500, 0.227021)( -4.122000, 0.223994)( -4.108500, 0.221160)( -4.095000, 0.218504)( -4.081500, 0.216011)( -4.068000, 0.213670)( -4.054500, 0.211469)( -4.041000, 0.209400)( -4.027500, 0.207454)( -4.014000, 0.205622)( -4.000500, 0.203898)( -3.987000, 0.202276)( -3.973500, 0.200749)( -3.960000, 0.199314)( -3.946500, 0.197965)( -3.933000, 0.196697)( -3.919500, 0.195509)( -3.906000, 0.194395)( -3.892500, 0.193353)( -3.879000, 0.192380)( -3.865500, 0.191474)( -3.852000, 0.190633)( -3.838500, 0.189853)( -3.825000, 0.189135)( -3.811500, 0.188475)( -3.798000, 0.187873)( -3.784500, 0.187327)( -3.771000, 0.186837)( -3.757500, 0.186401)( -3.744000, 0.186020)( -3.730500, 0.185691)( -3.717000, 0.185416)( -3.703500, 0.185194)( -3.690000, 0.185025)( -3.676500, 0.184908)( -3.663000, 0.184846)( -3.649500, 0.184837)( -3.636000, 0.184883)( -3.622500, 0.184984)( -3.609000, 0.185142)( -3.595500, 0.185359)( -3.582000, 0.185635)( -3.568500, 0.185973)( -3.555000, 0.186375)( -3.541500, 0.186844)( -3.528000, 0.187383)( -3.514500, 0.187994)( -3.501000, 0.188683)( -3.487500, 0.189453)( -3.474000, 0.190310)( -3.460500, 0.191260)( -3.447000, 0.192310)( -3.433500, 0.193468)( -3.420000, 0.194743)( -3.406500, 0.196147)( -3.393000, 0.197692)( -3.379500, 0.199395)( -3.366000, 0.201272)( -3.352500, 0.203348)( -3.339000, 0.205648)( -3.325500, 0.208206)( -3.312000, 0.211063)( -3.298500, 0.214274)( -3.285000, 0.217905)( -3.271500, 0.222047)( -3.258000, 0.226821)( -3.244500, 0.232397)( -3.231000, 0.239020)( -3.217500, 0.247056)( -3.204000, 0.257090)( -3.202200, 0.258625)( -3.200400, 0.260214)( -3.198600, 0.261862)( -3.196800, 0.263572)( -3.195000, 0.265348)( -3.193200, 0.267193)( -3.191400, 0.269112)( -3.189600, 0.271111)( -3.187800, 0.273193)( -3.186000, 0.275364)( -3.184200, 0.277630)( -3.182400, 0.279997)( -3.180600, 0.282470)( -3.178800, 0.285055)( -3.177000, 0.287755)( -3.175200, 0.290574)( -3.173400, 0.293510)( -3.171600, 0.296559)( -3.169800, 0.299702)( -3.168000, 0.302906)( -3.166200, 0.306103)( -3.164400, 0.309162)( -3.162600, 0.311833)( -3.160800, 0.313607)( -3.159000, 0.313378)( -3.157200, 0.308452)( -3.155400, 0.290833)( -3.153600, 0.224388)( -3.151800, 0.000000)( -3.150000, 0.000000)
};
\draw[color = red, line width = .5] plot[smooth] coordinates{
( -2.998000, 1.093359)( -2.996000, 0.872851)( -2.994000, 0.766210)( -2.992000, 0.699155)( -2.990000, 0.651607)( -2.988000, 0.615448)( -2.986000, 0.586659)( -2.984000, 0.562977)( -2.982000, 0.543019)( -2.980000, 0.525878)( -2.978000, 0.510935)( -2.950000, 0.402307)( -2.925000, 0.360532)( -2.900000, 0.335193)( -2.875000, 0.317934)( -2.850000, 0.305406)( -2.825000, 0.295962)( -2.800000, 0.288686)( -2.775000, 0.283024)( -2.750000, 0.278618)( -2.725000, 0.275229)( -2.700000, 0.272691)( -2.675000, 0.270888)( -2.650000, 0.269738)( -2.625000, 0.269182)( -2.600000, 0.269184)( -2.575000, 0.269719)( -2.550000, 0.270777)( -2.525000, 0.272359)( -2.500000, 0.274478)( -2.475000, 0.277157)( -2.450000, 0.280432)( -2.425000, 0.284352)( -2.400000, 0.288983)( -2.375000, 0.294412)( -2.350000, 0.300750)( -2.325000, 0.308142)( -2.300000, 0.316779)( -2.275000, 0.326917)( -2.250000, 0.338898)( -2.225000, 0.353203)( -2.200000, 0.370518)( -2.175000, 0.391867)( -2.150000, 0.418855)( -2.125000, 0.454167)( -2.100000, 0.502734)( -2.075000, 0.574876)( -2.050000, 0.697399)( -2.025000, 0.977120)( -2.020000, 1.090445)( -2.010000, 1.536493)
};
\filldraw (-3.15,0) circle (.02);
\draw (-3.20,0) node[below]{$s$};
\draw[black] (2,0) -- (3.7,0);
\filldraw (2,0) circle (.02) node[below] {$a$};
\filldraw (2.7,0) circle (.02) node[below] {$0$};
\filldraw (3.7,0) circle (.02)  node[below] {$1$};

\draw[color = red, line width = .5] plot[smooth] coordinates{
( 2.007000, 2.129405)( 2.014000, 1.511201)( 2.021000, 1.238426)( 2.028000, 1.076487)( 2.035000, 0.966445)( 2.042000, 0.885574)( 2.049000, 0.823011)( 2.056000, 0.772822)( 2.063000, 0.731458)( 2.070000, 0.696648)( 2.077000, 0.666863)( 2.084000, 0.641032)( 2.091000, 0.618377)( 2.098000, 0.598321)( 2.105000, 0.580423)( 2.112000, 0.564340)( 2.119000, 0.549803)( 2.126000, 0.536593)( 2.133000, 0.524536)( 2.140000, 0.513487)( 2.147000, 0.503324)( 2.154000, 0.493948)( 2.161000, 0.485273)( 2.168000, 0.477228)( 2.175000, 0.469750)( 2.182000, 0.462785)( 2.189000, 0.456288)( 2.196000, 0.450217)( 2.203000, 0.444538)( 2.210000, 0.439220)( 2.217000, 0.434234)( 2.224000, 0.429558)( 2.231000, 0.425168)( 2.238000, 0.421047)( 2.245000, 0.417177)( 2.252000, 0.413542)( 2.259000, 0.410129)( 2.266000, 0.406925)( 2.273000, 0.403920)( 2.280000, 0.401103)( 2.287000, 0.398466)( 2.294000, 0.396000)( 2.301000, 0.393699)( 2.308000, 0.391555)( 2.315000, 0.389564)( 2.322000, 0.387719)( 2.329000, 0.386017)( 2.336000, 0.384454)( 2.343000, 0.383025)( 2.350000, 0.381729)( 2.357000, 0.380562)( 2.364000, 0.379523)( 2.371000, 0.378610)( 2.378000, 0.377822)( 2.385000, 0.377158)( 2.392000, 0.376618)( 2.399000, 0.376202)( 2.406000, 0.375911)( 2.413000, 0.375745)( 2.420000, 0.375707)( 2.427000, 0.375798)( 2.434000, 0.376021)( 2.441000, 0.376379)( 2.448000, 0.376876)( 2.455000, 0.377516)( 2.462000, 0.378304)( 2.469000, 0.379247)( 2.476000, 0.380352)( 2.483000, 0.381625)( 2.490000, 0.383078)( 2.497000, 0.384719)( 2.504000, 0.386561)( 2.511000, 0.388617)( 2.518000, 0.390904)( 2.525000, 0.393439)( 2.532000, 0.396244)( 2.539000, 0.399341)( 2.546000, 0.402761)( 2.553000, 0.406536)( 2.560000, 0.410704)( 2.567000, 0.415312)( 2.574000, 0.420414)( 2.581000, 0.426078)( 2.588000, 0.432383)( 2.595000, 0.439430)( 2.602000, 0.447340)( 2.609000, 0.456271)( 2.616000, 0.466423)( 2.623000, 0.478059)( 2.630000, 0.491530)( 2.637000, 0.507320)( 2.644000, 0.526117)( 2.651000, 0.548941)( 2.658000, 0.577377)( 2.665000, 0.614069)( 2.672000, 0.663833)( 2.679000, 0.736706)( 2.686000, 0.858626)( 2.693000, 1.132654)
};
\draw[color = red, line width = .5] plot[smooth] coordinates{
( 2.800900, 0.000000)( 2.801800, 0.396577)( 2.802700, 0.488039)( 2.803600, 0.511062)( 2.804500, 0.516025)( 2.805400, 0.514424)( 2.806300, 0.510129)( 2.807200, 0.504713)( 2.808100, 0.498885)( 2.809000, 0.492987)( 2.818000, 0.444607)( 2.827000, 0.413284)( 2.836000, 0.391274)( 2.845000, 0.374695)( 2.854000, 0.361608)( 2.863000, 0.350930)( 2.872000, 0.342006)( 2.881000, 0.334409)( 2.890000, 0.327848)( 2.899000, 0.322120)( 2.908000, 0.317071)( 2.917000, 0.312589)( 2.926000, 0.308587)( 2.935000, 0.304996)( 2.944000, 0.301762)( 2.953000, 0.298840)( 2.962000, 0.296195)( 2.971000, 0.293797)( 2.980000, 0.291620)( 2.989000, 0.289644)( 2.998000, 0.287852)( 3.007000, 0.286227)( 3.016000, 0.284758)( 3.025000, 0.283432)( 3.034000, 0.282241)( 3.043000, 0.281176)( 3.052000, 0.280230)( 3.061000, 0.279395)( 3.070000, 0.278668)( 3.079000, 0.278043)( 3.088000, 0.277516)( 3.097000, 0.277083)( 3.106000, 0.276741)( 3.115000, 0.276488)( 3.124000, 0.276321)( 3.133000, 0.276239)( 3.142000, 0.276239)( 3.151000, 0.276322)( 3.160000, 0.276485)( 3.169000, 0.276729)( 3.178000, 0.277053)( 3.187000, 0.277456)( 3.196000, 0.277940)( 3.205000, 0.278504)( 3.214000, 0.279150)( 3.223000, 0.279877)( 3.232000, 0.280688)( 3.241000, 0.281584)( 3.250000, 0.282566)( 3.259000, 0.283637)( 3.268000, 0.284799)( 3.277000, 0.286054)( 3.286000, 0.287406)( 3.295000, 0.288858)( 3.304000, 0.290414)( 3.313000, 0.292077)( 3.322000, 0.293853)( 3.331000, 0.295745)( 3.340000, 0.297761)( 3.349000, 0.299906)( 3.358000, 0.302187)( 3.367000, 0.304612)( 3.376000, 0.307188)( 3.385000, 0.309926)( 3.394000, 0.312836)( 3.403000, 0.315929)( 3.412000, 0.319218)( 3.421000, 0.322718)( 3.430000, 0.326444)( 3.439000, 0.330415)( 3.448000, 0.334652)( 3.457000, 0.339177)( 3.466000, 0.344016)( 3.475000, 0.349201)( 3.484000, 0.354765)( 3.493000, 0.360748)( 3.502000, 0.367196)( 3.511000, 0.374163)( 3.520000, 0.381712)( 3.529000, 0.389919)( 3.538000, 0.398871)( 3.547000, 0.408677)( 3.556000, 0.419468)( 3.565000, 0.431405)( 3.574000, 0.444687)( 3.583000, 0.459570)( 3.592000, 0.476381)( 3.601000, 0.495548)( 3.610000, 0.517649)( 3.619000, 0.543479)( 3.628000, 0.574173)( 3.637000, 0.611416)( 3.646000, 0.657844)( 3.655000, 0.717862)( 3.664000, 0.799535)( 3.673000, 0.919733)( 3.682000, 1.122214)( 3.691000, 1.581148)
};
\filldraw (2.8009,0) circle (.02);
\draw (2.85, 0) node[below] {$s$};
\end{tikzpicture}
\caption{Plot of the limiting particle densities in Angelesco ensembles with two weights on $[a,0]$ and $[0,1]$ 
in the cases $a < -1$, $a=-1$, and $-1 < a < 0$.  For clarity the size of the gap between 
$0$ and $s$ has been exaggerated in the figure. In the left and right plots the density has
a local maximum near $s$. Even though the plots may suggest otherwise, the density is in fact
real analytic there.}
\label{fig:187}
\end{figure}
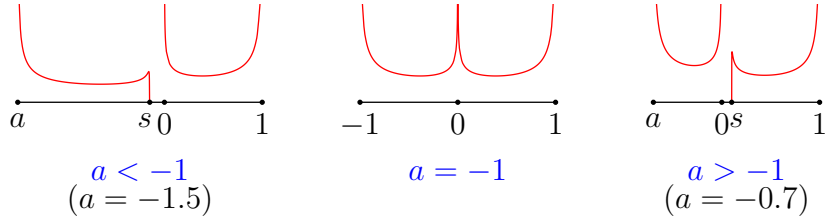

So for the remainder of the paper we take for $a < 0$
\begin{equation}
\Delta_{1} := [a, 0], \qquad \Delta_{2} := [0,1],
\end{equation} and as in \cite{deku11} we consider the following modified Jacobi weights:
\begin{equation} \label{eq:95}
\begin{aligned}
		 w_{1}(x) & := (x-a)^{\alpha} |x|^{\beta} h_{1}(x) \chi_{[a,0]}(x),  \\ 
		 w_{2}(x) & := x^{\beta}(1-x)^{\gamma} h_{2}(x) \chi_{[0,1]}(x), 
		 \end{aligned}
\end{equation}
with $\alpha, \beta, \gamma > -1$ and functions $h_{j}$ that are positive on $\Delta_{j}$ 
and analytic in a neighborhood of $\Delta_{}$ for $j=1,2$. Note that we use the same exponent
$\beta$ in both $w_1$ and $w_2$.

In \cite{deku11} we considered the weights \eqref{eq:95} and obtained a Mehler-Heine asymptotic 
formula for the MOP associated to $w_{1}$ and $w_{2}$ and $\vec{n} = (n,n)$ as $n \to \infty$.
See also \cite{tak09,tul09}.

We use the notation $K_{n,n}(\cdot, \cdot; a)$ to denote the dependence of 
the correlation kernel \eqref{eq:Kn} on the parameter $a$.  The main result is then as follows:

\begin{theorem} \label{thm:145}
Let $K_{n,n}$ be the correlation kernel of the MOP ensemble with
weights \eqref{eq:95} and diagonal multi-index $\vec{n} = (n,n)$. Let $\tau, x, y \in \mathbb{R}$ with $y \neq 0$.
Then as $n \rightarrow + \infty$ we have:
\begin{equation}
	\frac{1}{\sqrt{2} n^{\frac{3}{2}}} K_{n,n} \left( \frac{x}{\sqrt{2} n^{\frac{3}{2}}}, \frac{y}{\sqrt{2} n^{\frac{3}{2}}}; -1 + 			\frac{\sqrt{2}\tau}{n^{\frac{1}{2}}} \right) 
	= \mathbb{K}^{\Ang}(x,y;\tau) + \mathcal{O}\left( \frac{ y^{\beta}}{n^{\frac{1}{6}}} 	\right) \label{eq:150}
\end{equation}
for some limiting kernel $\mathbb{K}^{\Ang}(x,y;\tau)$ that will be described below in Propositions
\ref{prop:169} and \ref{prop:183}.
The $\mathcal{O}$-term in \eqref{eq:150} is uniform for $x$ in a bounded set 
and $y$ in a compact subset in $\mathbb{R} \setminus \{ 0 \}$.
\end{theorem}

The limiting kernel $\mathbb{K}^{\Ang}$ is referred to as the \textit{Angelesco kernel}.  
It has an expression in terms of the Angelesco model parametrix $\Psi$ which we 
will describe in Section \ref{sec:356}.  More explicit expressions exist, 
one involving a certain pairing of two analytic functions $q_{0}, r_{0}$ and 
one in terms of a double contour integral.

\begin{proposition}\label{prop:169}
Define functions $q_{0}$ and $r_{0}$ by
\begin{align} \label{eq:181}
	q_{0}(x) & = q_{0}(x;\tau) := 
	\frac{1}{2\pi i} |x|^{\beta +2} \int_{\Gamma_{0}} 
	t^{-\beta -3} e^{\frac{\tau x}{t} - \frac{ x^{2}}{2t^{2}} + t } \, dt, && x \in \mathbb{R}, \\
 \label{eq:184}
	r_{0}(y) & = r_{0}(y;\tau) 
	:= \frac{1}{2\pi i} |y|^{-\beta -1} \int_{\widehat{\Gamma}_{0}} 
	s^{\beta} e^{-\frac{\tau y }{s} + \frac{y^{2}}{2s^{2}} - s} \,ds, && y \in \mathbb{R}\setminus \{0\}.
\end{align} 
Here the contours $\Gamma_{0}$ and $\widehat{\Gamma}_{0}$ are shown 
in Figure~\ref{fig:197}, and we choose the principal branches for the fractional powers 
$t \mapsto t^{-\beta -3}$ and $s \mapsto s^{\beta}$.  
Then we can write $\mathbb{K}^{\Ang}(x,y;\tau)$ as a pairing of $q_{0}(x)$ and $r_{0}(y)$ in the following way:
\begin{equation} \label{eq:187}
		\mathbb{K}^{\Ang}(x,y;\tau) = 
		\left| \frac{y}{x} \right|^{\beta} \left[ \begin{array}{c} 
		yr_{0}(y) q_{0}''(x) - \left( (\beta +1) r_{0}(y) + y r_{0}'(y) \right) q_{0}'(x) \\ 
		+ \left( yr_{0}''(y) + (\beta +2)r_{0}'(y) - \tau r_{0}(y) \right) q_{0}(x) \end{array} \right].
\end{equation}
\end{proposition}

\begin{figure}[t]
\begin{center}
\begin{tikzpicture}[scale = 1.5, line width = .66]
\filldraw (0,0) circle (.03) node[above left]{$0$};
\draw [dashed] (0,0) -- (-1.5,0);
\draw [postaction = decorate, decoration = {markings, mark = at position 0 with {\arrow{stealth};}}] (.4,0) ellipse (.4 and .6);
\draw [postaction = decorate, decoration = {markings, mark = at position .5 with {\arrow{stealth};}}] (.5,.8) arc (90:-90:.5 and .8);
\draw (.5,.8) .. controls (0,.8) and (-.5,.5) .. (-1,.5);
\draw (.5,-.8) .. controls (0,-.8) and (-.5,-.5) .. (-1,-.5);
\draw (.9,-.5) node [right] {$\Gamma_{0}$};
\draw (-.1,-.1) node[below] {$\widehat{\Gamma}_{0}$};
\draw [dotted] (-1,.5) -- (-1.3,.5);
\draw [dotted] (-1,-.5) -- (-1.3,-.5);
\end{tikzpicture}
\caption{The contours $\Gamma_{0}$ and $\widehat{\Gamma}_{0}$ appearing in the integral formulas \eqref{eq:181}, \eqref{eq:184},
and \eqref{eq:235} for the Angelesco kernel $\mathbb{K}^{\Ang}$. The dashed line denotes the 
branch cut of $t^{\beta}$ and $s^{\beta}$.\label{fig:197}}
\end{center}
\end{figure}
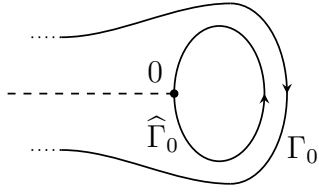

From the expressions for $q_{0}$ and $r_{0}$ it is clear that we have
\begin{equation}
q_{0}(x;\tau) = q_{0}(-x;-\tau), \qquad r_{0}(y;\tau) = r_{0}(-y; -\tau).
\end{equation}
Hence the Angelesco kernel satisfies the following (expected) symmetry:
\begin{equation}
\mathbb{K}^{\Ang}(x,y;\tau) = \mathbb{K}^{\Ang} (-x,-y;-\tau).
\end{equation}

\begin{proposition} \label{prop:183}
Let the contours $\Gamma_{0}$ and $\widehat{\Gamma}_{0}$ be again as in Figure~\ref{fig:197}, 
and take $x,y \in \mathbb{R}
\setminus \{ 0 \}$, $\tau \in \mathbb{R}$.  Then we have the following double integral formula for $\mathbb{K}^{\Ang}(x,y;\tau)$:
\begin{equation} \label{eq:235}
\mathbb{K}^{\Ang}(x,y;\tau) = \frac{ \sign (y) }{(2\pi i)^{2}} \int_{t \in \Gamma_{0}} \int_{s \in \widehat{\Gamma}_{0}} \frac{s^{\beta}}{t^{\beta }} \frac{1}{xs-yt} \frac{ e^{\frac{ \tau x}{t} - \frac{ x^{2}}{2t^{2}} +t} }{e^{\frac{ \tau y}{s} - \frac{y^{2}}{2s^{2}} +s }} \,ds \, dt.
\end{equation}
Again main branches are used for the fractional powers $t \mapsto t^{\beta}$ and $s \mapsto s^{\beta}$.
\end{proposition}

\begin{remark}
In \cite{kmfw09}, \cite{kmfw11} the authors consider a model of $n$ non-intersecting squared Bessel paths, conditioned 
to start at time $t= 0$ in the point $x = a >0$, and to end at time $t = 1$ at $x = 0$, see Figure~\ref{fig:246}.  
For each fixed time $t$ the particles form a MOP ensemble with respect to orthogonality weights involving 
the modified Bessel functions $I_{\beta}$ and $I_{\beta +1}$ for some $\beta > -1$.

\begin{figure}[t]
\centering
\includegraphics[scale = .75, trim = 2cm 10cm 2cm 11.5cm, clip=true]{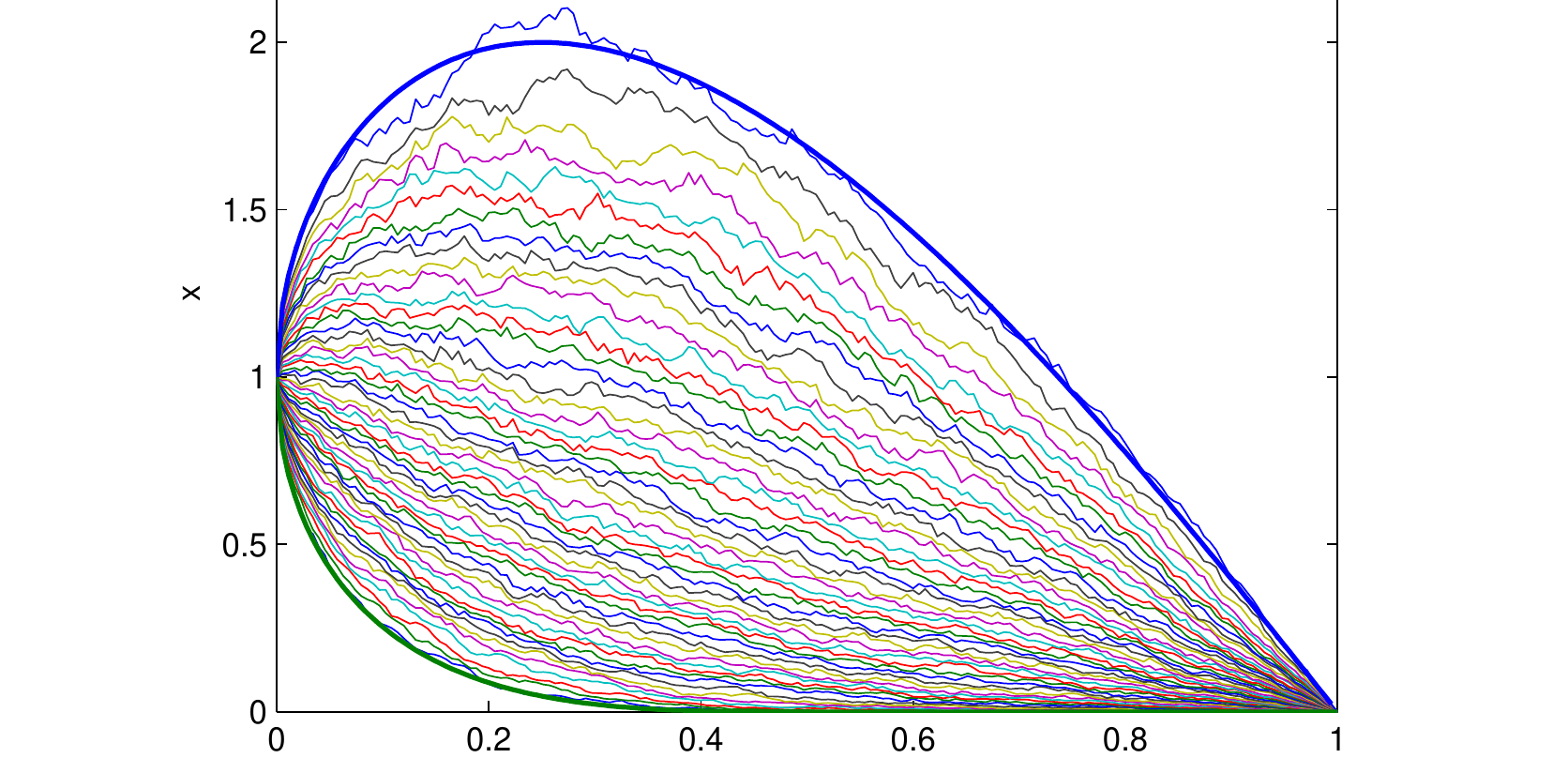}
\caption{Numerical simulation of $50$ non-intersecting squared Bessel paths.  
On the horizontal axis, time is running from $0$ to $1$. The vertical axis has the position $x$. The figure
is taken from \cite{kmfw09}.}
\label{fig:246}
\end{figure}

Associated to each fixed time $t$ one has a limiting mean particle distribution as $n \to \infty$.  
For small values of $t$ the particles are supported on an interval in $\mathbb{R}^{+}$, away from $0$.  
On the other hand, for values of $t$ near $1$ the point $0$ is contained in the support, meaning that 
particles stick to the hard wall in $0$.  There is then a critical time $t = t^{*} \in (0,1)$ at which 
the gap between the particles and the wall in $0$ closes.  This phase transition is similar to the phase transition in Angelesco ensembles.

By taking a double scaling limit around $x = 0$ and $t = t^{*}$ a new kind of universal limiting kernel, 
$K^{crit}$, was obtained \cite[equation (1.18)]{kmfw11}.  The function $K^{crit}$ depends on
 two positive position variables $x$ and $y$, a double scaling parameter $\tau$, and the parameter 
 $\beta$ appearing in the modified Bessel functions.  It turns out that this critical kernel is 
 almost the same as the Angelesco kernel $\mathbb{K}^{\Ang}$.  By applying a simple substitution 
 to the double integral formula for $K^{crit}$ \cite[equation (1.19)]{kmfw11} one obtains the following identity:
\begin{equation} \label{eq:2480}
\mathbb{K}^{\Ang}_{\beta}(x,y;\tau) = \frac{ y^{\beta}}{x^{\beta}}K^{crit}_{\beta}(y,x;\tau), \quad x, y > 0,
\end{equation}
where we used the notation $\mathbb{K}^{\Ang}_{\beta}$ to indicate the dependence of the kernel on the parameter $\beta$.

Although the kernels are different this identity implies an equality at the level 
of the correlation functions.  Hence, after proper identification of the parameters, 
the critical behavior in Angelesco ensembles coincides with the critical behavior in 
the model of non-intersecting squared Bessel paths.  A significant difference between the 
two models is that in the squared Bessel paths model the particles remain on the positive half-line, 
while in our model  they are located on both sides of $0$ and there is interaction between the two groups.
\end{remark}

The proof of Theorem \ref{thm:145} follows from the application of the Deift/Zhou steepest 
descent method on the Riemann-Hilbert (RH) problem for MOPs.  In Section \ref{sec:301} we 
recall the RH formulation for the weights $w_{1}$ and $w_{2}$ \eqref{eq:95} and we give a formula 
for the MOP kernel in terms of the RH problem for MOPs.  Also we recall the RH problem 
for the Angelesco local model parametrix $\Psi$ which is taken from \cite{deku11}.

The Deift/Zhou method consists of a sequence of invertible transformations reducing the 
original RH problem into a normalized RH problem, for which uniform estimates can be made.  
The analysis is the same as in \cite{deku11}.  Instead of giving full 
details we will, in Section \ref{sec:260}, give a quick overview of the transformations 
and describe their effect on the expression for the MOP kernel.

In Section \ref{sec:288} we shall use the series of transformations to show that the 
limit in Theorem \ref{thm:145} indeed holds, and we give expressions for the Angelesco 
kernel $\mathbb{K}^{\Ang}$ in terms of the Angelesco parametrix $\Psi$.

Finally, Section \ref{sec:785} contains the proofs of the two explicit expressions for 
$\mathbb{K}^{\Ang}$ stated in Proposition \ref{prop:169} and Proposition \ref{prop:183}.

\section{Riemann-Hilbert problem}
\label{sec:301}

This section recalls the Riemann-Hilbert (RH) characterization of the relevant multiple 
orthogonal polynomials and the Angelesco local model parametrix.

\subsection{Multiple orthogonal polynomials}
\label{sec:303}
 
We state here the RH problem for MOPs  \cite{vgk01} with respect to the system of weights $(w_{1}, w_{2})$ 
given in \eqref{eq:95} and the multi-index $(n_{1},n_{2}) \in \mathbb{N}^{2}$.  
The endpoint conditions are as in \cite{kmvv04}.

We look for a $3 \times 3$ matrix-valued function $Y$ such that:
\begin{itemize}
\item $Y$ is defined and analytic on $\mathbb{C} \setminus [a,1]$.
\item $Y$ has continuous boundary values $Y_{\pm}$ on $(a,0)$ and $(0,1)$ and they satisfy the jump relation $Y_{+} = Y_{-} J_{Y}$ with $J_{Y}$ given by
\begin{equation} \label{eq231}
J_{Y}(x) = \begin{pmatrix} 1 & w_{1}(x) & w_{2}(x) \\ 0 & 1 & 0 \\ 0 & 0 & 1 \end{pmatrix}, \quad x \in (a,0) \cup (0,1).
\end{equation}
Recall that $w_1(x) \equiv 0$ outside of $[a,0]$ and $w_2(x) \equiv 0$ outside of $[0,1]$.
\item As $z \rightarrow \infty$ we have
\begin{equation}
Y(z) = \left( I + \mathcal{O}\left(\frac{1}{z} \right)\right) \begin{pmatrix} z^{n_1+n_2} & 0 & 0 \\ 0 & z^{-n_1} & 0 \\ 0 & 0 & z^{-n_2} \end{pmatrix}.
\end{equation}
\item $Y$ has the following behavior at the endpoints of the intervals:
\begin{align} \label{eq:240}
& Y(z) = \mathcal{O} \begin{pmatrix} 1 & \epsilon(z) & 1 \\ 1 & \epsilon(z) & 1 \\ 1 & \epsilon(z) & 1 \end{pmatrix}, \textrm{ as } z  \to a, \notag \\
& \qquad \qquad  \textrm{where } \epsilon(z) = \begin{cases}  (z-a)^{\alpha} & \textrm{ if } \alpha < 0, \\ \log (z-a) & \textrm{ if } \alpha =0, \\ 1 & \textrm{ if } \alpha > 0. \end{cases}
\end{align}
\begin{align}\label{eq:244}
& Y(z) = \mathcal{O}\begin{pmatrix}  1 & 1 & \epsilon(z) \\ 1 & 1 & \epsilon(z) \\ 1 & 1 & \epsilon(z) \end{pmatrix}, \textrm{ as } z  \to 1, \notag \\
& \qquad  \qquad \textrm{where } \epsilon(z) = \begin{cases}  (z-1)^{\gamma} & \textrm{ if } \gamma < 0, \\ \log(z-1) & \textrm{ if } \gamma =0, \\ 1 & \textrm{ if } \gamma > 0. \end{cases}
\end{align}
\begin{align} \label{eq:248}
& Y(z) = \mathcal{O} \begin{pmatrix}  1 & \epsilon(z) & \epsilon(z) \\ 1 & \epsilon(z) & \epsilon(z) \\ 1 & \epsilon(z) & \epsilon(z) \end{pmatrix}, \textrm{ as } z  \to 0, \notag \\
& \qquad \qquad \textrm{where } \epsilon(z) = \begin{cases}  z^{\beta} & \textrm{ if } \beta < 0, \\ \log z & \textrm{ if } \beta =0, \\ 1 & \textrm{ if } \beta > 0. \end{cases}
\end{align}
The $\mathcal{O}$-symbol is to be taken entry-wise.
\end{itemize}

This RH problem has a unique solution in terms of MOP with respect to the weights $w_{1}$ and $w_{2}$.
In particular we have that $Y_{11}(z)$ is a monic polynomial of degree $n_1 + n_2$ that
satisfies the multiple orthogonality conditions
\[ \int x^k Y_{11}(x) w_j(x) dx = 0, \qquad k=0, \mathellipsis, n_j-1, \quad j=1,2. \]
For the full expression of $Y$ we refer to \cite{deku11}.  
Most important for the present paper is the following formula for the MOP correlation kernel in terms of $Y$,
see \cite{blku04a,daku04}:
\begin{equation}
	K_{n_{1},n_{2}}(x,y) = 
	\frac{1}{2\pi i(x-y)} \begin{pmatrix} 0 & w_{1}(y) & w_{2}(y) \end{pmatrix} Y_{+}(y)^{-1} 
	Y_{+}(x) 	\begin{pmatrix} 1 \\ 0 \\ 0 \\ \end{pmatrix}. \label{eq:254}
\end{equation}
This formula is valid for $x \neq y \in \mathbb{R}$ with $y \notin \{ a,0,1\}$.  As before it 
is understood that $w_{1}(y) \equiv 0$ for $y$ outside $(a,0)$, and $w_{2}(y) \equiv 0$ for 
$y$ outside $(0,1)$.  The formula \eqref{eq:254} can be extended to the case 
$x= y \notin \{a,0,1 \}$ by l'H\^{o}pital's rule.

\subsection{Angelesco local model parametrix}
\label{sec:356}

At a crucial step in the steepest descent analysis we need to do a local analysis at 
the point $0$.  This step involves certain special functions, which are combined 
into a $3\times 3$ matrix-valued function, called the Angelesco model parametrix 
$\Psi$ and which was introduced in \cite[Section 2.2]{deku11}.  
The function $\Psi$ depends on two parameters $\beta > -1$ 
and $\tau \in \mathbb{R}$.  Since $\beta$ is considered fixed we do not emphasize 
the dependence on $\beta$. We may write $\Psi(z;\tau)$ to emphasize the dependence on $\tau$.

The function $\Psi$ has the following RH characterization:
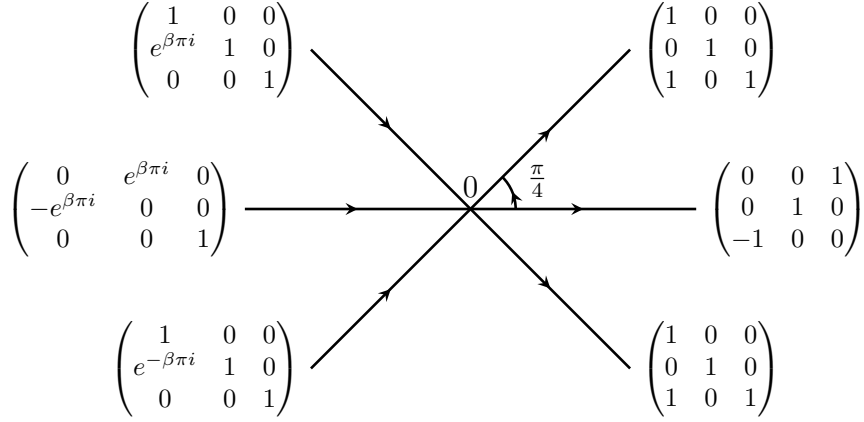
\begin{figure}[t]
\begin{center}
\begin{tikzpicture}[scale = 3,line width = 1, decoration = {markings, mark = at position .5 with {\arrow{stealth};}} ]
\draw [postaction = decorate] (-1,0) -- (0,0);
\draw [postaction = decorate] (0,0) -- (1,0);
\draw [postaction = decorate] (-.707106781,-.707106781) -- (0,0);
\draw [postaction = decorate] (0,0) -- (.707106781,.707106781);
\draw [postaction = decorate] (-.707106781,.707106781) -- (0,0);
\draw [postaction = decorate] (0,0) -- (.707106781,-.707106781);
\draw (0,0) node[above]{$0$};
\draw [postaction = decorate] (.2,0) arc  (0:45:.2);
\draw (.2,0) node[above right]{$\frac{\pi}{4}$};
\draw (-1,0) node[left,font=\footnotesize]{ $\begin{pmatrix} 0 & e^{\beta \pi i} & 0 \\ -e^{\beta \pi i} & 0 & 0 \\ 0 & 0 & 1 \end{pmatrix}$};
\draw (-.707106781,.707106781) node[left,font=\footnotesize]{$\begin{pmatrix} 1 & 0 & 0 \\ e^{\beta \pi i} & 1 & 0 \\ 0 & 0 & 1 \end{pmatrix}$};
\draw (.707106781,.707106781) node[right,font=\footnotesize]{$\begin{pmatrix} 1 & 0 & 0 \\ 0 & 1 & 0 \\ 1 & 0 & 1 \end{pmatrix}$};
\draw (1,0) node[right,font=\footnotesize]{ $\begin{pmatrix} 0 & 0 & 1 \\ 0 & 1 & 0 \\ -1 & 0 & 0 \end{pmatrix}$};
\draw (.707106781,-.707106781) node[right,font=\footnotesize]{$\begin{pmatrix} 1 & 0 & 0 \\ 0 & 1 & 0 \\ 1 & 0 & 1 \end{pmatrix}$};
\draw (-.707106781,-.707106781) node[left,font=\footnotesize]{$\begin{pmatrix} 1 & 0 & 0 \\ e^{-\beta \pi i} & 1 & 0 \\ 0 & 0 & 1 \end{pmatrix}$};
\end{tikzpicture}
\caption{The contour $\Sigma_{\Psi}$ and the jump matrices of $\Psi$.}
\label{figurejumpspsi}
\end{center}
\end{figure}

\begin{itemize}
\item $\Psi$ is defined and analytic on $\mathbb{C} \setminus \Sigma_{\Psi}$ where $\Sigma_{\Psi}$ is 
a contour consisting of six oriented rays through the origin, as shown in Figure \ref{figurejumpspsi}.
\item $\Psi$ has continuous boundary values on $\Sigma_{\Psi} \setminus \{0\}$ that satisfy the jump
condition
\[ \Psi_{+}(z) = \Psi_{-}(z) J_{\Psi}(z) \qquad z \in \Sigma_{\Psi} \setminus \{0\}, \]
where the jump matrices $J_{\Psi}$ are also given in Figure \ref{figurejumpspsi}.
\item Denote $\omega := e^{2\pi i/3}$.  As $z \to \infty$ with $\pm \im z > 0$, we have
\begin{equation}
    \Psi(z) =  \sqrt{ \frac{ 2\pi}{3}} e^{\frac{\tau^{2}}{6}} z^{\frac{\beta}{3}}
    \begin{pmatrix} z^{\frac{1}{3}} & 0 & 0 \\ 0 & 1 & 0 \\ 0 & 0 & z^{-\frac{1}{3}}  \end{pmatrix}
    L_{\pm}
    \left(I + \mathcal{O}\left(z^{-\frac{1}{3}}\right)\right) B_{\pm} e^{\Theta(z; \tau)}, 
    \label{asymptoticexpansionpsi}\end{equation}
where $L_{\pm}$, $B_{\pm}$ and $\Theta(z;\tau)$ are defined by
\begin{equation}  \label{definitionomega}
    \begin{aligned}
    L_+ & := \begin{pmatrix} - \omega^{2} & 1 & \omega \\ 1 & -1 & -1 \\ -\omega & 1 & \omega^{2}   \end{pmatrix},
    & \quad B_+ & := \begin{pmatrix} e^{\frac{\beta\pi i}{3}} & 0 &0 \\ 0 & 1 & 0 \\ 0 & 0 & e^{-\frac{\beta \pi i}{3}} \end{pmatrix}, \\
    L_- & := \begin{pmatrix} \omega & 1 & \omega^{2} \\ -1 & -1 & -1 \\ \omega^{2} & 1 & \omega \end{pmatrix},
    & \quad B_- & :=
     \begin{pmatrix} e^{-\frac{\beta \pi i}{3}} & 0 & 0 \\ 0 & 1 & 0 \\ 0 & 0 & e^{\frac{\beta \pi i}{3}} \end{pmatrix},
      \end{aligned}
\end{equation}
\begin{equation} \label{definitionTheta}
    \Theta(z; \tau) :=
    \begin{cases} \diag \left( \theta_{1}(z;\tau), \theta_{3}(z;\tau), \theta_{2}(z;\tau) \right) & \textrm{ for } \im z > 0, \\
    \diag \left( \theta_{2}(z;\tau), \theta_{3}(z;\tau), \theta_{1}(z;\tau) \right) & \textrm{ for } \im z < 0, \end{cases}
    \end{equation}
and the $\theta_{k}$ are defined by
\begin{equation}
    \theta_{k}(z;\tau) := -\frac{3}{2} \omega^{k} z^{\frac{2}{3}} - \tau \omega^{2k} z^{\frac{1}{3}}
    \qquad \textrm{ for } k =1,2,3. \label{definitionthetak}
\end{equation}
The expansion \eqref{asymptoticexpansionpsi} for $\Psi(z)$ as $z \to \infty$  
is valid uniformly for $\tau$ in a bounded set.
\item Denote by $s_{1}, s_{2}$ the sectors
\begin{align}
s_{1} & := \left \{ z \in \mathbb{C} \, \middle | \, \arg z \in \left( -\frac{3\pi }{4}, - \frac{\pi}{4} \right) \cup \left(\frac{\pi}{4}, \frac{3\pi}{4}\right) \right \}, \notag \\
s_{2} & := \left \{ z \in \mathbb{C} \, \middle | \, \arg z \in \left( -\frac{\pi }{4}, \frac{\pi}{4} \right) \cup \left(\frac{3\pi}{4}, \frac{5\pi}{4}\right) \right \}. \label{eq:4790}
\end{align}

Around $0$ we then have the following estimate:
\begin{multline}
	\Psi(z) = \mathcal{O}\begin{pmatrix} 
	\epsilon_{1}(z) & \epsilon_{2}(z) & \epsilon_{2}(z) \\ \epsilon_{1}(z) & \epsilon_{2}(z) & \epsilon_{2}(z) \\ 
	\epsilon_{1}(z) & \epsilon_{2}(z) & \epsilon_{2}(z) \end{pmatrix}  \textrm{ as }z \rightarrow 0,  \\
\textrm{ with } (\epsilon_{1}(z), \epsilon_{2}(z)) = 
	\begin{cases} (z^{\beta}, z^{\beta}) & \textrm{ for } \beta < 0 , \\ (1, \log z) & \textrm{ for } \beta = 0, z \in s_{1}, \\ 
	(\log z, \log z) & \textrm{ for } \beta = 0, z \in s_{2}, \\ 
	(z^{\beta},1) & \textrm{ for } \beta > 0, z \in s_{1}, \\ 
	(1,1) & \textrm{ for } \beta > 0, z \in s_{2}.  \end{cases} \label{eq479}
\end{multline}
\end{itemize}

Note that the parameter $\beta$ appears in the jump condition and in the behavior around $0$. 
The dependence on $\tau$ is only in the asymptotic condition as $z \rightarrow \infty$.

The RH problem for $\Psi$ has a unique solution.  It is constructed out of solutions of the 
third order linear differential equation
\begin{equation}
zq'''(z) - \beta q''(z) -  \tau q'(z) + q(z) = 0. \label{diffeqq}
\end{equation}
This differential equation \eqref{diffeqq} has solutions in the form of contour integrals
\begin{equation}
q(z) = \int_{\Gamma} t^{-\beta - 3} e^{\frac{\tau}{t} - \frac{1}{2t^{2}} + zt} \,dt , 
\end{equation}
where $\Gamma$ is an appropriate contour so that the integrand  vanishes at the endpoints of the contour $\Gamma$.  Define three contours $\Gamma_{1}, \Gamma_{2}$ and $\Gamma_{3}$ as in Figure \ref{figurecontoursgamma},
and define for $z$ with $\re z > 0$
\begin{equation}
q_{j}(z) := \int_{\Gamma_{j}} t^{-\beta - 3} e^{\frac{\tau}{t} - \frac{1}{2t^{2}} + zt} \,dt, \qquad j=1,2,3, \label{integralrepresentationsq}
\end{equation}
where we choose the branch of $t^{-\beta -3}$ with a cut on the positive real axis, i.e.,
\[ t^{-\beta-3} = |t|^{-\beta-3} e^{(-\beta-3) i \arg t}, \qquad 0 < \arg t < 2 \pi. \]

The integrals \eqref{integralrepresentationsq} only converge for $z$ with
$\re z > 0$, but the functions $q_{j}$ can be continued analytically using contour deformations.
Branch points for the $q$-functions are $0$ and $\infty$ and we take the analytic
continuation to $\mathbb C \setminus (-\infty,0]$, thus with a branch cut on the negative real axis.

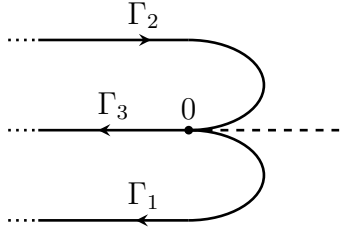
\begin{figure}[t]
\begin{center}
\begin{tikzpicture}[scale = 2,line width = 1]
\filldraw (0,0) circle (.02) node[above]{$0$};
\draw [dashed] (0,0) -- (1,0);
\draw [postaction = decorate, decoration = {markings, mark = at position .5 with {\arrowreversed{stealth};}}] (-1,0) -- (0,0);
\draw (-.5,0) node[above]{$\Gamma_{3}$};
\draw [dotted] (-1.2,0) -- (-1,0);
\draw (0,0) arc (-90:90:.5 and .3);
\draw [postaction = decorate, decoration = {markings, mark = at position .75 with {\arrow{stealth};}}] (-1,.6) -- (0,.6);
\draw [dotted] (-1.2,.6) -- (-1,.6);
\draw (-.3,.6) node[above]{$\Gamma_{2}$};
\draw (0,0) arc (90:-90:.5 and .3);
\draw [postaction = decorate, decoration = {markings, mark = at position .75 with {\arrowreversed{stealth};}}] (-1,-.6) -- (0,-.6);
\draw [dotted] (-1.2,-.6) -- (-1,-.6);
\draw (-.3,-.6) node[above] {$\Gamma_{1}$};
\end{tikzpicture}
\caption{The contours $\Gamma_{1},\Gamma_{2}$ and $\Gamma_{3}$ in the $t$-plane.
    The dashed line denotes the cut of $t^{-\beta -3}$.}
\label{figurecontoursgamma}
\end{center}
\end{figure}

In the upper half-plane the unique solution for $\Psi$ is then given by, see \cite{deku11} for more details,
\begin{equation}
\Psi : = \begin{cases} \begin{pmatrix} e^{2\beta \pi i} q_{1} & e^{\beta \pi i} q_{3}  & q_{2} \\ e^{2\beta \pi i} q_{1}' & e^{\beta \pi i} q_{3}'  & q_{2}' \\ e^{2\beta \pi i} q_{1}'' & e^{\beta \pi i} q_{3}''  & q_{2}'' \end{pmatrix}, & 0 < \arg z < \frac{\pi}{4}, \\ \begin{pmatrix} e^{2\beta \pi i} q_{1} + q_{2} & e^{\beta \pi i} q_{3}  & q_{2} \\ e^{2\beta \pi i} q_{1}' + q_{2}' & e^{\beta \pi i} q_{3}'  & q_{2}' \\ e^{2\beta \pi i} q_{1}'' + q_{2}'' & e^{\beta \pi i} q_{3}''  & q_{2}'' \end{pmatrix}, & \frac{\pi}{4} < \arg z < \frac{3\pi}{4}, \\ \begin{pmatrix} e^{2\beta \pi i} q_{1} + q_{2} -e^{2\beta \pi i} q_{3} & e^{\beta \pi i} q_{3}  & q_{2} \\ e^{2\beta \pi i} q_{1}' + q_{2}' -e^{2\beta \pi i} q_{3}' & e^{\beta \pi i} q_{3}'  & q_{2}' \\ e^{2\beta \pi i} q_{1}'' + q_{2}'' -e^{2\beta \pi i} q_{3}'' & e^{\beta \pi i} q_{3}''  & q_{2}'' \end{pmatrix}, & \frac{3\pi}{4} < \arg z < \pi, \end{cases} \label{definitionpsi1}
\end{equation}
and in the lower half-plane by
\begin{equation}
\Psi := \begin{cases} \begin{pmatrix} q_{2} & e^{\beta \pi i} q_{3} & -e^{2\beta \pi i} q_{1} \\ 
q_{2}' & e^{\beta \pi i} q_{3}' & -e^{2\beta \pi i} q_{1}' \\ 
q_{2}'' & e^{\beta \pi i} q_{3}'' & -e^{2\beta \pi i} q_{1}'' \end{pmatrix}, & -\frac{\pi}{4} < \arg z < 0, \\ 
\begin{pmatrix} q_{2} + e^{2\beta \pi i} q_{1} & e^{\beta \pi i} q_{3} & -e^{2\beta \pi i} q_{1} \\ 
q_{2}' + e^{2\beta \pi i} q_{1}' & e^{\beta \pi i} q_{3}' & -e^{2\beta \pi i} q_{1}' \\ 
q_{2}'' + e^{2\beta \pi i} q_{1}''& e^{\beta \pi i} q_{3}'' & -e^{2\beta \pi i} q_{1}'' \end{pmatrix}, 
& -\frac{3\pi}{4} < \arg z <- \frac{\pi}{4}, \\ 
\begin{pmatrix} e^{2\beta \pi i} q_{1} + q_{2} + q_{3} & e^{\beta \pi i} q_{3} & -e^{2\beta \pi i} q_{1} \\ 
e^{2\beta \pi i} q_{1}' + q_{2}' + q_{3}' & e^{\beta \pi i} q_{3}' & -e^{2\beta \pi i} q_{1}' \\ 
e^{2\beta \pi i} q_{1}'' + q_{2}'' + q_{3}'' & e^{\beta \pi i} q_{3}'' & -e^{2\beta \pi i} q_{1}''  \end{pmatrix}, 
& -\pi < \arg z < - \frac{3\pi}{4}. \end{cases} \label{definitionpsi2}
\end{equation}

\section{Steepest descent transformations}
\label{sec:260}

In this section we apply the Deift/Zhou steepest descent method to the RH problem for $Y$.  
For details of the various transformations we refer to \cite{deku11}.

The Deift/Zhou method consists of a number of invertible transformations reducing the original matrix-valued
function $Y$ to a function $R$ that is uniformly close to the identity matrix.  Of main importance 
is the effect of these transformations on the expression for the MOP kernel $K_{n_{1},n_{2}}(x,y)$, see \eqref{eq:254}, 
with diagonal multi-indices:
\begin{equation}
n_{1} = n = n_{2}.
\end{equation}

The first transformation $Y \mapsto T$ normalizes the RH problem for $Y$ at infinity: the new matrix-valued
function $T$ tends to the identity matrix at infinity.  The transformation uses functions $g_{1}$ and 
$g_{2}$ derived from a modified equilibrium problem with logarithmic potentials.  These $g$-functions
 are analytic on $\mathbb{C} \setminus (-\infty, 1]$. See \cite[Section 4.1]{deku11} for precise
 details on the transformation $Y \mapsto T$. For the MOP kernel we obtain
 from \eqref{eq:254} and formula (4.4) in \cite{deku11}
\begin{multline} \label{eq:KninT}
	K_{n,n}(x,y) = \frac{1}{2\pi i(x-y)} 
	\begin{pmatrix}  0  \\ w_{1}(y) e^{2n(g_{1,+}(y) + l_{1} + \pi i/2)} \\ w_{2}(y) e^{2n(g_{2,+}(y) + l_{2})} 
	\end{pmatrix}^{\transpose} \\ 
 \times T_{+}(y)^{-1} T_{+}(x) \begin{pmatrix} e^{2n(g_{1,+}(x) + g_{2,+}(x))} \\ 0 \\ 0 \end{pmatrix}.
\end{multline}
The constants $l_{1}$ and $l_{2}$ appear in the Euler-Lagrange variational conditions related to the equilibrium problem.

The second transformation $T \mapsto S$ is the opening of the lenses, see \cite[Section 4.2]{deku11}.  
In this transformation the rapidly oscillating jump matrices along $(a,0)$ and $(0,1)$ 
that appeared after the transformation $Y \mapsto T$ are turned into constant jump matrices, with the side effect of creating exponentially small jump matrices on the lips of the lenses.  
The expression for this transformation depends on the location of $z$.  
From now on we restrict ourselves to positive $x$ and $y$, since in the cases $x < 0$ or $y < 0$ the expressions for $K_{n,n}(x,y)$ are different but similar.  Remark that we have $w_{1}(y) = 0$ for $y > 0$.  By \eqref{eq:KninT} and 
formulas (4.14) and (4.5) in \cite{deku11}, we then get for $x, y \in (0,1)$, (note that $w_1(y) = 0$):
\begin{multline} \label{eq:KninS}
	K_{n,n}(x,y) = \frac{1}{2\pi i(x-y)} 
	\begin{pmatrix} -e^{-2n( g_{1,+}(y) + g_{2,+}(y))} \\ 0 \\ w_{2}(y) e^{2n(g_{2,+}(y) + l_{2})} \end{pmatrix}^{\transpose} \\
	\times S_{+}(y)^{-1} S_{+}(x) 
	\begin{pmatrix} e^{2n(g_{1,+}(x) + g_{2,+}(x))} \\ 0 \\ w_{2}(x)^{-1} e^{-2n(g_{2,+}(x) + l_{2})} \end{pmatrix}.
\end{multline}

The next two transformations $S \mapsto R_{0}$ and $R_{0} \mapsto R$ are applied after 
constructing local and global parametrices.  These parametrices are built as approximations 
to the RH problem for $S$, and one can think of $R_{0}$ as the approximation error.  
The approximation around $0$ is not good enough for the further analysis, whence 
the need for a very last transformation $R_{0} \mapsto R$. We refer to \cite[Section 4.4]{deku11}
for precise details and formulas.

The local parametrix $P_0$ around $0$ is constructed on a disk $U_{0}$ centered in zero with 
shrinking radius $n^{-1/2}$, see formulas (4.68) and (4.70) of \cite{deku11}.  
It involves the Angelesco model parametrix $\Psi$ described in 
Section~\ref{sec:356}.  Assume that $x$ and $y$ are positive and belong to the disk $U_{0}$.
Then $S = R_0 P_0$ and inserting this in \eqref{eq:KninS} we obtain after some calculations
\begin{multline} 
K_{n,n}(x,y) = \frac{1}{2\pi i(x-y)} \left|\frac{y^{\beta}}{x^{\beta}} \right|
	\begin{pmatrix} -1 & 0 & 1 \end{pmatrix} \Psi_{+} \left( n^{\frac{3}{2}} f(y); n^{\frac{1}{2}} \tau(y) \right)^{-1} \\
	\times E_{n}(y)^{-1} R_0(y)^{-1} R_0(x) E_n(x)  
		\Psi_{+} \left( n^{\frac{3}{2}} f(x); n^{\frac{1}{2}} \tau(x) \right) 
		\begin{pmatrix}  1 \\ 0 \\ 1 \end{pmatrix}. \label{eq:284}
\end{multline}

% \begin{multline}
% K_{n,n}(x,y) = \frac{1}{2\pi i(x-y)} \frac{y^{\beta}}{x^{\beta}} 
%	\begin{pmatrix} -1 & 0 & 1 \end{pmatrix} \Psi_{+} \left( n^{\frac{3}{2}} f(y); n^{\frac{1}{2}} \tau(y) \right)^{-1} \\
%	\times E_{n}(y)^{-1} \left( I + Z_{n}(y) \right) R(y)^{-1} R(x) \left( I - Z_{n}(x) \right)  \\
%		\times E_{n}(x) \Psi_{+} \left( n^{\frac{3}{2}} f(x); n^{\frac{1}{2}} \tau(x) \right) 
%		\begin{pmatrix}  1 \\ 0 \\ 1 \end{pmatrix}. \label{eq:284}
% \end{multline}

For $x < 0$ or $y < 0$ the formula \eqref{eq:284} also holds, but with different row and column vectors, 
as given in the following table:
\begin{equation} \label{eq:550}
\begin{array}{c|c|c} & x < 0 & x> 0 \\ \hline 
	y < 0 & 
	\begin{pmatrix} -1 & 1 & 0 \end{pmatrix}, \begin{pmatrix}  1 \\ 1 \\ 0 \end{pmatrix} & 
	\begin{pmatrix} -1 & 1 & 0 \end{pmatrix}, \begin{pmatrix}  1 \\ 0 \\ 1 \end{pmatrix} \\
	\hline y > 0 & 
	\begin{pmatrix} -1 & 0 & 1 \end{pmatrix}, \begin{pmatrix}  1 \\ 1 \\ 0 \end{pmatrix} & 
	\begin{pmatrix} -1 & 0 & 1 \end{pmatrix}, \begin{pmatrix}  1 \\ 0 \\ 1 \end{pmatrix} 
	\end{array}
\end{equation}

The functions $f$ and $\tau$ appearing in \eqref{eq:284} were defined in \cite{deku11}.  Since we 
do not need to know their full expressions we will restrict ourselves to recalling 
their most important properties. 

These functions are both analytic in a neighborhood of $0$, with $f(0) = 0$.  They depend on $a$, and as $a \rightarrow -1$ one has:
\begin{equation} \label{eq:530}
f'(0;a) = \sqrt{2} + \mathcal{O}(a+1), \qquad \tau(0;a) = \frac{1}{\sqrt{2}} (a+1) + \mathcal{O}(a+1)^{2},
\end{equation}
see \cite[formula (4.76)]{deku11}.
From this it easily follows that:
\begin{lemma}
Let $x \in \mathbb R \setminus \{0\}$ and $ \tau \in \mathbb R$. 
Put 
\begin{equation} x_n := \frac{x}{\sqrt{2} n^{\frac{3}{2}}}, \quad a_n := -1 + \frac{ \sqrt{2} \tau}{n^{\frac{1}{2}}}.\end{equation}
Then as $n \rightarrow \infty$ we have
\begin{align} \label{eq:531}
n^{\frac{3}{2}} f(x_{n};a_{n}) & = x + \mathcal{O}\left(n^{-\frac{1}{2}}\right), \\
n^{\frac{1}{2}} \tau(x_{n};a_{n}) & = \tau + \mathcal{O}\left(n^{-\frac{1}{2}}\right). \label{eq:806}
\end{align}
The $\mathcal{O}$-term is uniform for $x$ in a bounded set.  
\end{lemma}
\begin{proof}
This follows easily from \eqref{eq:530} and the fact that $f(z; a)$ and $\tau(z, a)$ 
converge uniformly as $a \rightarrow -1$ to $f(z;-1)$ and $\tau(z;-1)$.
\end{proof}

Also we have the following technical estimate:
\begin{proposition}
\label{prop:338}
Let $x, y \in \mathbb R \setminus \{0\}$, $\tau \in \mathbb R$ and put
\begin{equation} \label{eq:xnynan} 
	x_n :=  \frac{x}{\sqrt{2} n^{\frac{3}{2}}}, \quad 
	y_n := \frac{y}{\sqrt{2} n^{\frac{3}{2}}}, \quad a_n := -1 + \frac{ \sqrt{2} \tau}{n^{\frac{1}{2}}}. 
	\end{equation}
Then as $n \rightarrow \infty$ we have
\begin{equation}
	E_{n}(y_{n};a_n)^{-1} R_0(y_{n})^{-1} R_0(x_{n}) E_{n}(x_{n};a_n) = 
	I + \mathcal{O}\left(n^{-\frac{1}{6}}\right). \label{eq:3400}
\end{equation}
The $\mathcal{O}$-term is uniform for $x$ and $y$ in bounded sets.
\end{proposition}

Assuming this proposition we can finish the proof of Theorem \ref{thm:145} as follows.
Let $x, y > 0$ and let $x_n, y_n, a_n$ be as in \eqref{eq:xnynan}. Then for $n$ large
enough we have that $x_n, y_n \in U_0$. Using \eqref{eq:531}, \eqref{eq:806} and \eqref{eq:3400}
in \eqref{eq:284} we then obtain:
\begin{multline} \label{eq:479}
	\frac{1}{\sqrt{2}n^{\frac{3}{2}}} K_{n,n}(x_{n},y_{n};a_{n}) = \\ 
		\frac{1}{2\pi i(x-y)} \left| \frac{ y}{x} \right|^{\beta} 
			\begin{pmatrix} -1 & 0 & 1 \end{pmatrix}  \Psi(y;\tau)^{-1} 
	  \Psi(x;\tau) \begin{pmatrix} 1 \\ 0 \\ 1 \end{pmatrix} + 
	\mathcal{O}\left( \frac{y^{\beta}}{ n^{\frac{1}{6}}} \right).
\end{multline}
This proves Theorem \ref{thm:145} for $x, y > 0$.  The right-hand side of \eqref{eq:479} 
also provides us with an expression for $\mathbb{K}^{\Ang}(x,y; \tau)$ in terms of the Angelesco parametrix $\Psi$.  
For $x< 0$ or $y < 0$ the proof is completely similar and \eqref{eq:479} holds, 
but with different row and column vectors, as given in \eqref{eq:550}.

\section{Proof of Proposition  \ref{prop:338}}
\label{sec:288}

\begin{proof}

The proof of Proposition \ref{prop:338} is rather technical and requires deeper 
knowledge of the structure of the matrix functions $R_{0}$ and $E_{n}$.%  We omit some of the details of the proof.

From \cite[Section 5.2]{deku11} we have
\begin{equation} \label{eq:Rz} 
	R_0(z) = R(z) (I - V_n(z;a)), \qquad z \in U_0,  
	\end{equation}
where, see formula (5.13) of \cite{deku11},
\begin{equation} \label{eq:Rzestimate} 
	R(z) =  I + \mathcal{O}(n^{-1/6}) \qquad \text{as } n \to \infty, 
	\end{equation}
uniformly in $z$.
The function $V_n$ in \eqref{eq:Rz} is given by
\begin{equation} V_n(z;a)  = Z_n(z;a) - \frac{Z_n^{(0)}(a)}{z} \end{equation}
where $Z_n$ and $Z_{n}^{(0)}(a)$ play a role in the analysis of \cite{deku11}.  The function $Z_{n}$ is analytic on $U_{0} \backslash \{ 0\}$ with a simple pole in $0$, whose residue is $Z_{n}^{(0)}(a)$.  Hence $V_{n}$ is analytic on $U_{0}$.  Additionally, from \cite[equation (4.84)]{deku11} it follows that $V_{n}$ is of the form
\begin{equation}\label{eq:formVn}
V_{n} = \left( \begin{array}{c} \textrm{scalar analytic} \\ \textrm{function in } z \end{array} \right) \times D_{\infty}^{-1} 
\begin{pmatrix} \sqrt{2} i \\ 1 \\ -1 \end{pmatrix} \begin{pmatrix} \sqrt{2}i & 1 & -1 \end{pmatrix} D_{\infty},
\end{equation} for a certain constant diagonal matrix $D_{\infty}$.  As a consequence we have the following
basic property of $V_n$:
\begin{equation} \label{eq:Vnnilpotent}
V_n(x;a) V_n(y;a) = 0, \qquad \text{for every } x,y \in U_0,
\end{equation}

% We will
%not give their definitions here, but we note the following
%basic property of $V_n$ (see \cite[formula (4.90)]{deku11}):
%\begin{equation} \label{eq:Vnnilpotent}
%	V_n(x;a) V_n(y;a) = 0, \qquad \text{for every } x,y \in U_0,
%	\end{equation}

Additionally, $V_n(z;a_n)$ is uniformly bounded on the disk $U_0$.
Since $U_0$ has radius $n^{-1/2}$ this easily implies that $V_n'(z_n;a_n) = \mathcal{O}(n^{1/2})$
if $z_n = \mathcal{O}(n^{-3/2})$, and therefore
\begin{equation} \label{eq:Vndifference} 
	V_n(x_n;a_n) - V_n(y_n;a_n) = \mathcal{O}(n^{-1}).
	\end{equation}
From \eqref{eq:Rz} and the fact that $x_n = \mathcal{O}(n^{-3/2})$, $y_n = \mathcal{O}(n^{-3/2})$, 
it follows that
\begin{equation} \label{eq:620}
	R(y_{n})^{-1} R(x_{n}) = I + \mathcal{O}\left(n^{-\frac{7}{6}}\right).
\end{equation}
Combining \eqref{eq:620}, \eqref{eq:Rzestimate}, \eqref{eq:Vnnilpotent}, 
and \eqref{eq:Vndifference}	 we obtain
\begin{equation} \label{eq:R0ratio} 
	R_0(y_n)^{-1} R_0(x_n) = I + \mathcal{O}(n^{-1}). 
	\end{equation}

The matrix-valued function $E_{n}$ defined in formula (4.80) of \cite{deku11} is analytic and invertible 
on a fixed neighborhood of $0$.  It depends on both $a$ and $n$.  It takes the form
\begin{equation} \label{eq:535}
E_{n}(z;a) = n^{-\frac{\beta}{2}} e^{-\frac{ n \tau(z;a)^{2}}{6}} \widetilde{E}(z;a) 
	\begin{pmatrix}  n^{-\frac{1}{2}} & 0 & 0 \\ 0 & 1 & 0 \\ 0 & 0 & n^{\frac{1}{2}} \end{pmatrix},
\end{equation} 
where $\widetilde{E}$ is independent of $n$, and $\widetilde{E}$ converges uniformly as $a \rightarrow -1$.
Then  
\begin{equation} \label{eq:631}
	\widetilde{E}(y_{n};a_n)^{-1} \widetilde{E}(x_{n};a_n) = I + \mathcal{O}\left( n^{-\frac{3}{2}} \right),
\end{equation} 
and combining with \eqref{eq:535} yields
\begin{equation} \label{eq:634}
E_{n}(y_{n};a_n)^{-1} E_{n}(x_{n};a_n) = I + \mathcal{O}\left( n^{-\frac{1}{2}} \right).
\end{equation}

It also follows from \eqref{eq:535} that 
\begin{equation} \label{eq:638}
	n^{\frac{\beta}{2}} E_{n}(x_{n};a_{n}) = \mathcal{O}\left( n^{\frac{1}{2}}\right), 
	\quad n^{-\frac{\beta}{2}} E_{n}(y_{n};a_{n})^{-1} = \mathcal{O} \left(n^{\frac{1}{2}} \right).
\end{equation}

A straightforward combination of  \eqref{eq:R0ratio}, \eqref{eq:634} and \eqref{eq:638},
would yield an $\mathcal{O}(1)$ bound for the left-hand side of \eqref{eq:3400}.  
We need to go a little deeper into the structure of $E_n$ in order to improve the bound.

From the computations in Section 6 of \cite{deku11} it follows that $E_{n}(x_{n};a_{n})$ has the following 
behavior as $n\rightarrow \infty$:
\begin{multline} \label{eq:Enxn}
E_{n}(x_{n};a_{n}) = C_{1} n^{\frac{1-\beta}{2}} D_{\infty}^{-1} 
\begin{pmatrix}  \sqrt{2}i \\ 1 \\ -1 \end{pmatrix} \begin{pmatrix} 0 & 0 & 1 \end{pmatrix} 
\left(I + \mathcal{O}\left(n^{-\frac{1}{2}}\right)\right),
\end{multline} 
where $C_{1}$ is some non-zero constant, see \cite[equation (6.12)]{deku11}.  
From similar computations one finds for some non-zero $C_{2}$:
\begin{equation} \label{eq:Enyninverse}
E_{n}(y_{n};a_{n})^{-1} = C_{2} n^{\frac{\beta +1}{2}} \begin{pmatrix} 1 \\ 0 \\ 0 \end{pmatrix} 
\begin{pmatrix} \sqrt{2}i & 1 & -1 \end{pmatrix} D_{\infty} \left(I + \mathcal{O}\left(n^{-\frac{1}{2}}\right)\right).
\end{equation}

From these expansions it follows that the leading terms of $E_{n}(y_{n};a_{n})$ and $E_{n}(x_{n};a_{n})$ 
cancel out with $V_{n}$, see \eqref{eq:formVn}.  
In particular it holds by \eqref{eq:formVn}, \eqref{eq:Vndifference}, \eqref{eq:Enxn} and \eqref{eq:Enyninverse} that
\begin{equation} \label{eq:639}
E_{n}(y_{n};a_n)^{-1} [V_{n}(y_{n};a_n) - V_{n}(x_{n};a_n)] E_{n}(x_{n}) = \mathcal{O}(n^{-1}) \textrm{ as } n \rightarrow \infty.
\end{equation}
Taking into account the extra estimate \eqref{eq:639} we can perform the following computation
(we do not write $a_n$, and we use \eqref{eq:Rz}, \eqref{eq:620}, \eqref{eq:Vnnilpotent}, \eqref{eq:634}, and the fact that $V_n(x_n)$ and $V_n(y_n)$
are uniformly bounded)
\begin{multline*}
 E_n(y_n)^{-1} R_0(y_n)^{-1} R_0(x_n) E_n(x_n) \\ 
 	\begin{aligned}
 & =  E_{n}(y_{n})^{-1} (I + V_{n}(y_{n})) R(y_{n})^{-1}R(x_{n}) (I- V_{n}(x_{n})) E_{n}(x_{n})  \\
 & = E_{n}(y_{n})^{-1} (I + V_{n}(y_{n})) 	\left[ I + \mathcal{O} \left( n^{-\frac{7}{6}} \right)\right] 
		(I- V_{n}(x_{n})) E_{n}(x_{n})  \\
 & = E_{n}(y_{n})^{-1} \left[ I + V_{n}(y_{n}) - V_{n}(x_{n}) + \mathcal{O}\left(n^{-\frac{7}{6}} \right)\right]  
	 E_{n}(x_{n}) \\
	& = I + \mathcal{O}\left(n^{-\frac{1}{6}}\right),
	\end{aligned}
\end{multline*}
which completes the proof of Proposition \ref{prop:338}.
\end{proof}

\section{Proof of Propositions \ref{prop:169} and \ref{prop:183}}
\label{sec:785}

In the previous section we derived expressions for $\mathbb{K}^{\Ang}$ in terms of $\Psi$, see 
right-hand side of \eqref{eq:479}.  
In the present section these expressions will be used to obtain the explicit formulas stated in 
Proposition \ref{prop:169} and Proposition \ref{prop:183}.
We follow here the computations in \cite[Section 10.2]{blku07}, see also \cite[Section 8.2]{kmfw11}.

As stated in Section \ref{sec:356} the entries of $\Psi$ are solutions and derivatives of solutions to the third order differential equation
\begin{equation}
zq'''(z) - \beta q''(z) - \tau q'(z) + q(z) = 0. \label{eq:488}
\end{equation}
The solution functions $q_{j}(z)$, $j=1,2,3$ have explicit expressions
\begin{equation}
q_{j}(z) = \int_{\Gamma_{j}} t^{-\beta - 3} e^{\frac{\tau}{t} - \frac{1}{2t^{2}} + zt} \,dt,
\end{equation} 
with contours $\Gamma_{j}$ as shown in Figure \ref{figurecontoursgamma}.

A first step towards the derivation of the double integral formula consists of finding similar explicit expressions for the entries 
of the inverse of $\Psi$.  This involves an \textit{adjoint differential equation} to the one for the $q$-functions \eqref{eq:488}, 
and an associated \textit{bilinear concomitant}, see e.g.~\cite{ber07}, \cite{ince44}.

The solutions of this adjoint differential equation have contour integral representations, and it takes some work to identify the appropriate contours. When the explicit expressions for the entries of $\Psi^{-1}$ are known we will be in a position to prove the double integral formula for $\mathbb{K}^{\Ang}$.

The adjoint differential equation reads
\begin{equation}
zr'''(z) + (\beta +3) r''(z) - \tau r'(z) - r(z)=0. \label{eq:501}
\end{equation}
The associated bilinear concomitant $\mathcal{B}$ is a differential operator defined on general twice differentiable functions $q$ and $r$:
\begin{align} \label{eq:506}
\mathcal{B}[q,r](x,y) & := yr(y) q''(x) - [(\beta +1) r(y) + yr'(y)] q'(x) \notag \\
& \qquad + [yr''(y) + (\beta +2) r'(y) - \tau r(y) ] q(x).
\end{align}
In order to ease the notation we introduce 3 operators $\mathcal{B}_{0}$, $\mathcal{B}_{1}$ and $\mathcal{B}_{2}$
\begin{equation} \begin{array}{rcl} \mathcal{B}_{0}r(z) & := & zr''(z) + (\beta + 2) r'(z) - \tau r(z), \\ \mathcal{B}_{1}r(z) & := & -zr'(z) - (\beta + 1 ) r(z), \\ \mathcal{B}_{2}r(z) &:=  &zr(z).\end{array} \label{eq:509} \end{equation}
So we can write the bilinear concomitant as
\begin{equation} \label{eq:511}
\mathcal{B}[q,r](x,y) = \mathcal{B}_{2}r(y) q''(x) + \mathcal{B}_{1}r(y) q'(x) + \mathcal{B}_{0}r(y) q(x) .
\end{equation}

The central property of this bilinear concomitant is that if $q$ solves the differential equation for $q$ \eqref{eq:488} and $r$ solves the differential equation for $r$ \eqref{eq:501}, then the bilinear concomitant evaluated on the diagonal $\mathcal{B}[q,r](x,x)$ is a constant.  This can be checked by taking the derivative of $\mathcal{B}[q,r](x,x)$ with respect to $x$ and plugging in the differential identities for $q$ and $r$.

The differential equation for $r$ \eqref{eq:488} allows solutions in the form of a contour integral: for an appropriate contour $\widehat{\Gamma}$ the function
\begin{equation} \label{eq:520}
r(z) := \int_{\widehat{\Gamma}} s^{\beta} e^{-\frac{\tau}{s} + \frac{1}{2s^{2}} - zs } \,ds
\end{equation} solves the differential equation \eqref{eq:501}.  By an appropriate contour is meant a contour such that the integrand vanishes exponentially fast at the endpoints.  In this case this implies that $\widehat{\Gamma}$ must end and start in $0$ or at $\infty$.

Now we construct solutions $r_{k}(z), k=1,2,3$ to \eqref{eq:501} such that
\begin{equation} \mathcal{B}[q_{j},r_{k}](z,z) = \delta_{jk} \quad j,k = 1,2,3, \label{concomitant01}\end{equation} where we put the cuts of the functions $r_{k}$, if any, on the negative real axis.  Using these relations we can find expressions 
for $\Psi^{-1}$.  In each sector of the complex plane $\Psi(z)$ is of the form
\begin{equation}
\Psi(z) = \begin{pmatrix} q_{1}(z) & q_{2}(z) & q_{3}(z) \\ q_{1}'(z) & q_{2}'(z) & q_{3}'(z) \\ q_{1}''(z) & q_{2}''(z) & q_{3}''(z) \end{pmatrix} \cdot C
\end{equation} for some invertible matrix $C$.  The inverse of $\Psi$ is then given by
\begin{equation}
\Psi(z)^{-1} = C^{-1} \cdot \begin{pmatrix} \mathcal{B}_{0} r_{1}(z) &  \mathcal{B}_{1} r_{1}(z) & \mathcal{B}_{2} r_{1}(z)  \\ \mathcal{B}_{0} r_{2}(z) &  \mathcal{B}_{1} r_{2}(z) & \mathcal{B}_{2} r_{2}(z) \\ \mathcal{B}_{0} r_{3}(z) &  \mathcal{B}_{1} r_{3}(z) & \mathcal{B}_{2} r_{3}(z) \end{pmatrix}.
\end{equation}
Indeed, the product $\Psi(z)^{-1} \Psi(z)$ then equals
\begin{equation}
C^{-1} \cdot  \begin{pmatrix} \mathcal{B}[q_{1},r_{1}](z,z) & \mathcal{B}[ q_{2},r_{1}](z,z) & \mathcal{B}[q_{3},r_{1}](z,z) \\ {} \mathcal{B}[q_{1},r_{2}](z,z) & \mathcal{B}[q_{2},r_{2}](z,z) & \mathcal{B}[q_{3},r_{2}](z,z) \\ {} \mathcal{B}[q_{1},r_{3}](z,z) & \mathcal{B}[q_{2}, r_{3}](z,z) & \mathcal{B}[q_{3}, r_{3}](z,z) \end{pmatrix}  \cdot C,
\end{equation} which by \eqref{concomitant01} is the identity matrix.

The functions $r_{k}(z)$ can be chosen to be of the form
\begin{equation} \label{eq:532}
r_{k}(z) = \frac{1}{2\pi i} \int_{\widehat{\Gamma}_{k}} s^{\beta} e^{-\frac{\tau}{s} + \frac{1}{2s^{2}} - zs } \,ds,
\end{equation}
with the contours $\widehat{\Gamma}_{k}$ determined by the following proposition:

\begin{proposition} \label{prop:537}
Let $\Omega \subset \mathbb{C}$ be a non-empty open set.  Assume that the following two contour integrals converge for $z \in \Omega$:
\begin{equation}
q(z) = \int_{\Gamma} t^{-\beta-3} e^{\frac{\tau}{t} - \frac{1}{2t^{2}} + zt } \, dt, \quad r(z) = \int_{\widehat{\Gamma}} s^{\beta} e^{-\frac{\tau}{s} + \frac{1}{2s^{2}} - zs } \, ds .\label{eq:539} \end{equation}
Also assume that the cuts of $t^{-\beta -3}$ and $s^{\beta}$ coincide, such that $t^{-\beta -3} s^{\beta} = t^{-3}$ for $t = s$.

If $\Gamma$ and $\widehat{\Gamma}$ do not intersect, or only intersect in $0$, then $\mathcal{B}[q,r](z,z) = 0$ for $z \in \Omega$.  If they however intersect once transversally at a non-zero point, and if $\Gamma$ meets $\widehat{\Gamma}$ on the negative side of $\widehat{\Gamma}$, then
\begin{equation}
\mathcal{B}[q,r](z,z) = 2 \pi i, \qquad z \in \Omega.
\end{equation}
\end{proposition}

\begin{proof}
In order to simplify the notations we introduce the two phase functions $\theta_{q}$ and $\theta_{r}$:
\begin{equation} \label{eq:545}
\theta_{q}(t) := \frac{\tau}{t} - \frac{1}{2t^{2}} + zt - \beta \log t,
\end{equation}
\begin{equation} \label{eq:548}
\theta_{r}(s) := -\frac{\tau}{s} + \frac{1}{2s^{2}} - zs + (\beta +3) \log s.
\end{equation}

Using the definition of the bilinear concomitant \eqref{eq:506} and the definitions of $q(z)$ and $r(z)$ \eqref{eq:539} we can write the concomitant $ \mathcal{B}[q,r](z,z)$ of $q$ and $r$ as a double integral:
\begin{equation}
\int_{t \in \Gamma} \int_{s \in \widehat{\Gamma}} \frac{ z(s^{2} + st + t^{2}) - (\beta +2)s - (\beta +1)t - \tau}{s^{3}t^{3}} e^{\theta_{q}(t)} e^{\theta_{r}(s)} \,dsdt. \label{eq:557}
\end{equation}
The fraction in the above integral can be split up as follows:
\begin{align} \label{eq:758}
& \frac{ z(s^{2} + st + t^{2}) - (\beta +2)s - (\beta +1)t - \tau}{s^{3}t^{3}} \notag \\
& \qquad = \frac{1}{t^{3} (t-s)} \frac{ \partial \theta_{r}(s)}{\partial s} + \frac{1}{s^{3} (t-s)} \frac{ \partial \theta_{q}(t)}{\partial t} - \frac{ s^{2} +st +t^{2}}{s^{3} t^{3} (t-s)}.
\end{align}

Suppose that the contours $\Gamma$ and $\widehat{\Gamma}$ do not intersect, or only intersect in $0$.  Then the three terms on the right-hand side of \eqref{eq:758} remain bounded and we can write:
\begin{align}\label{eq:558}
& \mathcal{B}[q,r](z,z) \notag \\
& \quad = \int_{t \in \Gamma} \int_{s \in \widehat{\Gamma}} \frac{ 1}{t^{3} (t-s)} \frac{ \partial \theta_{r}(s)}{\partial s} e^{\theta_{q}(t)} e^{\theta_{r}(s)} \, ds dt \notag \\
& \qquad +  \int_{t \in \Gamma} \int_{s \in \widehat{\Gamma}} \frac{ 1}{s^{3}(t-s)} \frac{\partial \theta_{q}(t)}{\partial t} e^{\theta_{q}(t)} e^{\theta_{r}(s)} \, ds dt \notag \\
& \qquad - \int_{t \in \Gamma} \int_{s \in \widehat{\Gamma}} \frac{s^{2} + st + t^{2}}{s^{3}t^{3} (t-s)}e^{\theta_{q}(t)} e^{\theta_{r}(s)} \, ds dt.
\end{align}

We can apply integration by parts to the first two integrals in \eqref{eq:558}.  The boundary terms vanish.  What remains is:
\begin{align}
\mathcal{B}[q,r](z,z) & = -\int_{t \in \Gamma} \int_{s \in \widehat{\Gamma}} \frac{ 1}{t^{3} (t-s)^{2}} e^{\theta_{q}(t)} e^{\theta_{r}(s)} \, ds dt \notag \\
& \quad + \int_{t \in \Gamma} \int_{s \in \widehat{\Gamma}} \frac{ 1}{s^{3}(t-s)^{2}} e^{\theta_{q}(t)} e^{\theta_{r}(s)} \, ds dt \notag \\
& \quad - \int_{t \in \Gamma} \int_{s \in \widehat{\Gamma}} \frac{t^{3} - s^{3}}{s^{3}t^{3} (t-s)^{2}}e^{\theta_{q}(t)} e^{\theta_{r}(s)} \, ds dt \notag \\
& = 0.
\end{align}

Now suppose that $\Gamma$ and $\widehat{\Gamma}$ intersect in a point $c \neq 0$, with $\Gamma$ going from the right side of $\widehat{\Gamma}$ to the left side.  Then the contours can be deformed such that $\Gamma$ contains the interval $[c-i\delta, c + i\delta]$ while $\widehat{\Gamma}$ contains $[c - \delta, c + \delta]$ for some $\delta > 0$, see Figure \ref{fig:611}.  Define a new contour $\widehat{\Gamma}^{\varepsilon}$ by $\widehat{\Gamma}^{\varepsilon} := \widehat{\Gamma} \setminus [c - \varepsilon, c + \varepsilon]$ with $0< \varepsilon < \delta$.  Then we have using \eqref{eq:557}:
\begin{equation}
\mathcal{B}[q,r](z,z) = \lim_{\varepsilon \rightarrow 0^{+}} \int_{t \in \Gamma} \int_{s \in \widehat{\Gamma}^{\varepsilon}} \left[ \begin{array}{c} z(s^{2} + st + t^{2}) - \tau  \\ + (\beta +2)s - (\beta +1)t \end{array} \right] \frac{ e^{\theta_{q}(t)} e^{\theta_{r}(s)}}{s^{3} t^{3}} \,dsdt.
\end{equation}

Since $\widehat{\Gamma}^{\varepsilon}$ and $\Gamma$ do not intersect we can use the expression in \eqref{eq:558} with $\widehat{\Gamma}$ replaced by $\widehat{\Gamma}^{\varepsilon}$ and apply integration by parts.  However, now the boundary term in $s$ remains:
\begin{equation} \mathcal{B}[q,r](z,z) = \lim_{\varepsilon \rightarrow 0^{+}} \int_{\Gamma} \frac{1}{t^{3}} \left[ \frac{1}{t-s} e^{\theta_{r}(s)} \right]^{s = c - \varepsilon}_{s = c + \varepsilon} e^{\theta_{q}(t)} \, dt .\end{equation}
By continuity we have $e^{\theta_{r}(c \pm \varepsilon)} = e^{\theta_{r}(c)} + \mathcal{O}(\varepsilon)$ as $\varepsilon \rightarrow 0$.  Hence we obtain
\begin{equation} \label{eq:584}
\mathcal{B}[q,r](z,z) =  \lim_{\varepsilon \rightarrow 0} \left( e^{\theta_{r}(c)} + \mathcal{O}(\varepsilon)\right) \int_{\Gamma} \frac{1 }{t^{3}} \left[\frac{1}{t-c+ \varepsilon} - \frac{1}{t-c-\varepsilon} \right]  e^{\theta_{q}(t)} \, dt.
\end{equation}

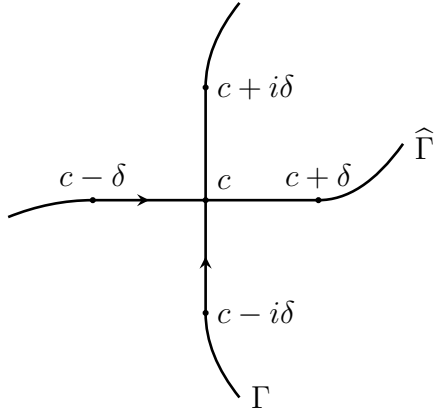
\begin{figure}[t]
\begin{center}
\begin{tikzpicture}[scale = .75, line width = 1,decoration = {markings, mark = at position .25 with {\arrow{stealth};}}]
\draw [postaction = decorate](-2, 0) -- (2,0) ;
\filldraw (-2,0) circle (.03) node[above]{$c- \delta$};
\filldraw (2,0) circle (.03) node[above]{$c+ \delta$};
\draw [postaction = decorate](0,-2) -- (0,2);
\filldraw (0, -2) circle (.03) node[right]{$c-i\delta$};
\filldraw (0, 2) circle (.03) node[right]{$c+i\delta$};
\filldraw (0,0) circle (.03) node[above right] {$c$};
\draw (2,0) .. controls (2.5,0) and (3,.3)..(3.5,1);
\draw (0,2) .. controls (0,2.5) and (.2, 3) .. (.6, 3.5);
\draw (3.5,1) node [right] {$\widehat{\Gamma}$};
\draw (0,-2) .. controls (0,-2.5) and (.2, -3) .. (.6, -3.5);
\draw (-2,0) .. controls (-2.5,0) and (-3,-.1)..(-3.5,-.3);
\draw (.6,-3.5) node [right] {$\Gamma$};
\end{tikzpicture}
\caption{The deformed contours $\Gamma$ containing the interval $[c- i\delta,  c +i \delta]$ and $\widehat{\Gamma}$ containing $[c-\delta, c + \delta]$.}
\label{fig:611}
\end{center}
\end{figure}

Finally we deform the segment $[c- i\delta, c + i \delta] \subset \Gamma$ into the union of a small circle around $c + \varepsilon$ and a semicircle with radius $\delta$ around $c$, see Figure \ref{fig:650}.  The small circle around $c + \varepsilon$ will give us a residue contribution, while the contribution of the semicircle and the remaining part of $\Gamma$ tends to zero as $\varepsilon \rightarrow 0^{+}$.  Remark that the small circle around $c + \varepsilon$ has clockwise orientation.  Then we have as $\varepsilon \rightarrow 0^{+}$:
\begin{equation} \int_{\Gamma} \frac{1 }{t^{3}} \left[\frac{1}{t-c + \varepsilon} - \frac{1}{t-c-\varepsilon} \right]  e^{\theta_{q}(t)} \, dt = 2 \pi i \frac{1}{(c + \varepsilon)^{3}} e^{\theta_{q}(c + \varepsilon)} +\mathcal{O}(\varepsilon) .\end{equation}

Taking the limit $\varepsilon \rightarrow 0^{+}$, and plugging the result into \eqref{eq:584} we get
\begin{equation} \mathcal{B}[q,r](z,z) = 2 \pi i \frac{1}{c^{3}} e^{\theta_{q}(c)}  e^{\theta_{r}(c)}   = 2 \pi i,
\end{equation}
where we used \eqref{eq:545} and \eqref{eq:548}.

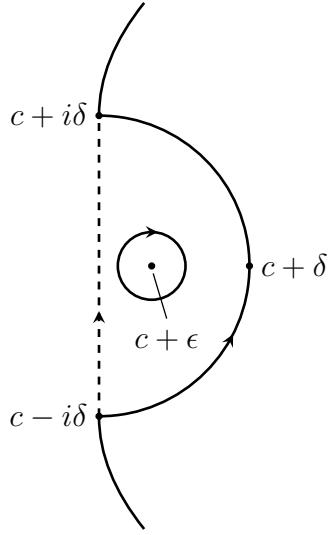
\begin{figure}[t]
\begin{center}
\begin{tikzpicture}[scale = 1, line width = 1]
\draw [line width = 1, dashed, postaction = decorate, decoration = {markings, mark = at position .35 with {\arrow{stealth};}}] (0,-2) -- (0,2);
\filldraw (.7,0) circle (.03);
\draw [line width = .5] (.72,-.1) -- (.9,-.7);
\draw (.9,-.7) node[below] {$c + \epsilon$};
\filldraw (2,0) circle (.03) node[right]{$c+ \delta$};
\draw [postaction = decorate, decoration = {markings, mark = at position .5 with {\arrow{stealth0};}}] (.7,-.45) arc (270:-90:.45);
\draw [postaction =decorate, decoration = {markings, mark = at position .35 with {\arrow{stealth};}}] (0,-2) arc (-90:90:2);
\draw (0,-2) .. controls (0,-2.5) and (.2, -3) .. (.6, -3.5);
\draw (0,2) .. controls (0,2.5) and (.2, 3) .. (.6, 3.5);
\filldraw (0,2) circle (.03) node[left] {$c + i \delta$};
\filldraw (0,-2) circle (.03) node[left] {$c - i \delta$};
\end{tikzpicture}
\caption{The dashed part of the contour $\Gamma$ is deformed into a semicircle of radius $\delta$ around $c$.  Additionally a residue is picked up in $c + \varepsilon$.  The contour $\widehat{\Gamma}$ is not shown here.}
\label{fig:650}
\end{center}
\end{figure}

\end{proof}

Choose contours $\widehat{\Gamma}_{k}$ as in Figure \ref{fig:656} and let the functions $r_{k}$ be defined by
\begin{equation}\label{eq:722}
r_{k}(z) := \frac{1}{2\pi i} \int_{\widehat{\Gamma}_{k}} s^{\beta} e^{-\frac{\tau}{s} + \frac{1}{2s^{2}} -zs } \,ds, \quad k = 1,2,3,
\end{equation} where $s^{\beta}$ has a cut on the positive real axis and $s^{\beta}_{+} = |s|^{\beta}$.  Then it follows from Proposition \ref{prop:537} that the concomitant conditions \eqref{concomitant01} are satisfied.  The integrals \eqref{eq:722} define analytic functions on $\mathbb{C}$ with possible branch points in $0$ and $\infty$ and we put the branch cut on the negative real axis.

\begin{figure}[t]
\begin{center}
\begin{tikzpicture}[scale = 2, line width = 1]
\filldraw (0,0) circle (.02) node[above right]{$0$};
\draw [dashed] (0,0) -- (1.5,0);
\draw [postaction = decorate, decoration = {markings, mark = at position .5 with {\arrowreversed{stealth};}}] (-.3,0) ellipse (.3 and .5);
\draw [postaction = decorate, decoration = {markings, mark = at position .5 with {\arrowreversed{stealth};}}] (0,0) .. controls (0,1) and (.6,.4) .. (1.2,.4);
\draw [postaction = decorate, decoration = {markings, mark = at position .5 with {\arrowreversed{stealth};}}] (0,0) .. controls (0,-1) and (.6,-.4) .. (1.2,-.4);
\draw [dotted] (1.2,.4) -- (1.4,.4);
\draw [dotted] (1.2,-.4) -- (1.4,-.4);
\draw (-.3,.5) node[above]{$\widehat{\Gamma}_{3}$};
\draw (.6,.4) node[left] {$\widehat{\Gamma}_{2}$};
\draw (.6,-.4) node[left]{$\widehat{\Gamma}_{1}$};
\end{tikzpicture}
\caption{The contours $\widehat{\Gamma}_{k}, k =1,2,3$.  The dashed line denotes the cut of $s^{\beta}$.}
\label{fig:656}
\end{center}
\end{figure}
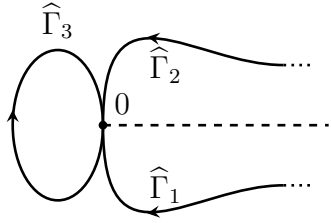

\begin{proof}[Proof of Proposition \ref{prop:169}]
Having found explicit expressions for the entries of $\Psi^{-1}$ we can return to the expression for $\mathbb{K}^{\Ang}$ given in \eqref{eq:479}.  The $1 \times 3$ and $3 \times 1$ matrices pick out certain linear combinations of the $q_{j}$ and $r_{k}$, and we obtain
\begin{equation} \label{eq:746}
\mathbb{K}^{\Ang}(x,y;\tau) = \frac{1}{(x-y)} \left| \frac{y}{x}\right|^{\beta} \mathcal{B} \left[ q_{0}, r_{0} \right](x,y),
\end{equation}
with $q_{0}$ and $r_{0}$ given by
\begin{equation} \label{eq:751}
q_{0}(x) := \frac{1}{2\pi i} \cdot \left\{ \begin{array}{cc} e^{2\beta \pi i} q_{1}(x) + q_{2}(x) & \textrm{ for } x > 0, \\ e^{\beta \pi i} q_{1,+}(x) + e^{-\beta \pi i} q_{2,+}(x) & \textrm{ for } x < 0, \end{array}\right.
\end{equation}
\begin{equation} \label{eq:754}
r_{0}(y) := \left\{ \begin{array}{cc} r_{2}(y)-e^{-2\beta \pi i} r_{1}(y) & \textrm{ for } y > 0, \\ e^{-\beta \pi i} r_{3}(y) & \textrm{ for } y < 0. \end{array}\right.
\end{equation}

The functions $q_{1}$ and $q_{2}$ have branch cuts on the negative real axis.  Hence we have to specify that we need the positive boundary value in \eqref{eq:751} in the case $x < 0$.  Using contour deformation and a simple substitution in the integral one can see that $q_{0}$ and $r_{0}$ are given by the contour integrals \eqref{eq:181} and \eqref{eq:184}.
\end{proof}

Finally we derive from \eqref{eq:746} the double integral formula.

\begin{proof}[Proof of Proposition \ref{prop:183}]

One can check that the derivatives of $q_{0}$ and $r_{0}$ have the following expressions:
\begin{equation}
q_{0}^{(n)}(x) = \frac{1}{2\pi i} x^{2-n} |x|^{\beta} \int_{\Gamma_{0}} t^{-\beta + n - 3} e^{\frac{ \tau x}{t} - \frac{ x^{2}}{2t^{2}} +t} \,dt,
\end{equation}
\begin{equation}
r_{0}^{(n)}(y) = \frac{ (-1)^{n}}{2 \pi i} y^{2-n} |y|^{-\beta -3} \int_{\widehat{\Gamma}_{0}} s^{\beta + n} e^{-\frac{ \tau y}{s} + \frac{y^{2}}{2s^{2}} -s } \,ds,
\end{equation}
for $n \in \mathbb{Z}^{+}$ and $x,y \in \mathbb{R} \setminus \{ 0 \}$.  The contours $\Gamma_{0}$ and $\widehat{\Gamma}_{0}$ are given in Figure \ref{fig:197}.  From \eqref{eq:746} and the definition of the bilinear concomitant \eqref{eq:506} it follows that we can write
\begin{align}
& \mathbb{K}^{\Ang}(x,y;\tau) = \frac{ 1}{4 \pi^{2} (y-x)|y|^{3}} \notag \\
& \qquad \cdot \int_{t \in \Gamma_{0}} \int_{s \in \widehat{\Gamma}_{0}} \frac{s^{\beta}}{t^{\beta +3}} \left[ \begin{array}{c} y^{3}t^{2} - (\beta +1) xy^{2}t + xy^{2}st \\ + x^{2}ys^{2} - (\beta +2) x^{2}ys - \tau x^{2}y^{2} \end{array} \right] \frac{ e^{\frac{ \tau x}{t} - \frac{ x^{2}}{2t^{2}} +t} }{e^{\frac{ \tau y}{s} - \frac{y^{2}}{2s^{2}} +s }} \,ds dt. \label{eq:906}
\end{align}

Introducing the phase functions $\eta_{q}$ and $\eta_{r}$
\begin{equation} \label{eq:911}
\eta_{q}(t) := - \beta \log t + \frac{ \tau x }{t} - \frac{ x^{2}}{2t^{2}} +t,
\end{equation}
\begin{equation} \label{eq:914}
\eta_{r}(s) := (\beta +3) \log s - \frac{ \tau y}{s} + \frac{ y^{2}}{2s^{2}} - s,
\end{equation} we can rewrite the double integral in \eqref{eq:906} as
\begin{equation} \label{eq:917}
\int_{t \in \Gamma_{0}} \int_{s \in \widehat{\Gamma}_{0}} \frac{1}{s^{3}t^{3}} \left[ \begin{array}{c} y^{3}t^{2} - (\beta +1) xy^{2}t + xy^{2}st \\ + x^{2}ys^{2} - (\beta +2) x^{2}ys - \tau x^{2}y^{2} \end{array} \right] e^{\eta_{q}(t)} e^{\eta_{r}(s)} \,ds dt.
\end{equation}

We have the following straightforward identity for the expression inside the square brackets:
\begin{align} \label{eq:330}
& y^{3}t^{2} -(\beta +1) xy^{2}t + xy^{2}st + x^{2}ys^{2} - (\beta +2) x^{2}ys - \tau x^{2} y^{2} \notag \\
& \qquad  = - \frac{ x^{3} ys^{3}}{xs-yt} \frac{\partial \eta_{r}(s)}{\partial s} - \frac{ xy^{3}t^{3} }{xs-yt}\frac{\partial \eta_{q}(t)}{\partial t} \notag \\
& \qquad + \frac{ x^{4}ys^{3} - xy^{4} t^{3}}{(xs-yt)^{2}} + y^{3}t^{3} \frac{ x-y}{xs-yt}.
\end{align}
The terms involving derivatives of the phase functions $\eta_{q}$ and $\eta_{r}$ can be simplified using integration by parts:
\begin{align}
& \int_{t \in \Gamma_{0}} \int_{s \in \widehat{\Gamma}_{0}} \frac{1}{s^{3}t^{3}} \frac{ x^{3}ys^{3} }{xs-yt} \frac{ \partial \eta_{r}(s)}{\partial s} e^{\eta_{q}(t)} e^{\eta_{r}(s)} \,ds dt \notag \\
& \qquad = \int_{t \in \Gamma_{0}} \int_{s \in \widehat{\Gamma}_{0}} \frac{1}{s^{3}t^{3}} \frac{ x^{4}ys^{3} }{(xs-yt)^{2}} e^{\eta_{q}(t)} e^{\eta_{r}(s)} \,ds dt, \label{eq:338} \\
& \int_{t \in \Gamma_{0}} \int_{s \in \widehat{\Gamma}_{0}} \frac{1}{s^{3}t^{3}} \frac{ xy^{3}t^{3} }{xs-yt} \frac{ \partial \eta_{q}(t)}{\partial t} e^{\eta_{q}(t)} e^{\eta_{r}(s)} \,ds dt \notag \\
& \qquad = - \int_{\Gamma_{0}} \int_{\widehat{\Gamma}_{0}} \frac{1}{s^{3}t^{3}} \frac{ xy^{4} t^{3} }{(xs-yt)^{2}} e^{\eta_{q}(t)} e^{\eta_{r}(s)} \,ds dt. \label{eq:340}
\end{align}

By plugging \eqref{eq:338} and \eqref{eq:340} together with \eqref{eq:330} into \eqref{eq:917} we find
\begin{align}
& \int_{t \in \Gamma_{0}} \int_{s \in \widehat{\Gamma}_{0}} \frac{1}{s^{3}t^{3}} \left[ \begin{array}{c} y^{3}t^{2} - (\beta +1) xy^{2}t + xy^{2}st \\ + x^{2}ys^{2} - (\beta +2) x^{2}ys - \tau x^{2}y^{2} \end{array} \right] e^{\eta_{q}(t)} e^{\eta_{r}(s)} \,ds dt \notag \\
& \qquad = (x-y) \int_{t \in \Gamma_{0}} \int_{s \in \widehat{\Gamma}_{0}} \frac{y^{3}}{s^{3}(xs-yt)} e^{\eta_{q}(t)} e^{\eta_{r}(s)} \,ds dt.
\end{align}
From \eqref{eq:906} we then conclude that $\mathbb{K}^{\Ang}(x,y;\tau)$ is given by
\begin{equation}
\mathbb{K}^{\Ang}(x,y;\tau) = \frac{ \sign (y) }{(2\pi i)^{2}} \int_{t \in \Gamma_{0}} \int_{s \in \widehat{\Gamma}_{0}} \frac{s^{\beta}}{t^{\beta }} \frac{1}{xs-yt} \frac{ e^{\frac{ \tau x}{t} - \frac{ x^{2}}{2t^{2}} +t} }{e^{\frac{ \tau y}{s} - \frac{y^{2}}{2s^{2}} +s }} \,ds dt.
\end{equation}

\end{proof}

\subsection{Acknowledgements}

The authors are supported by K.U. Leuven research grant OT/08/33,
and by the Belgian Interuniversity Attraction Pole P06/02.

The second author is also supported by FWO-Flanders projects G.0427.09 and G.0641.11,  
and by grant MTM2008-06689-C02-01 of the Spanish Ministry of Science and Innovation.

\end{document}